\documentclass[11pt,a4paper]{article}

\usepackage[truedimen,margin=34mm]{geometry} 

\usepackage{mathrsfs}
\usepackage{amssymb}
\usepackage{amsmath}
\usepackage{ascmac}
\usepackage{amsthm}
\usepackage[pdftex]{graphicx}
\usepackage[pdftex]{color}
\usepackage{natbib}
\usepackage{setspace}
\usepackage{color}

\usepackage{comment}
\usepackage{bm}
\usepackage{url}
\usepackage{multirow}

\usepackage{titlesec}
\titleformat*{\section}{\large\bfseries}
\titleformat*{\subsection}{\it}

\newtheorem{thm}{Theorem}
\newtheorem{lem}{Lemma}
\newtheorem{cor}{Corollary}
\newtheorem{prp}{Proposition}

\def\al{{\alpha}}
\def\be{{\beta}}
\def\ga{{\gamma}}
\def\de{{\delta}}
\def\ep{{\varepsilon}}

\def\si{{\sigma}}
\def\om{{\omega}}
\def\th{{\theta}}

\def\Ga{\Gamma}

\def\At{\widetilde{A}}
\def\Bt{\widetilde{B}}

\def\at{{\widetilde a}}
\def\bt{{\widetilde b}}
\def\mut{{\tilde \mu}}
\def\bet{{\tilde \be}}
\def\yt{{\tilde y}}
\def\pit{{\tilde \pi}}

\def\al{{\alpha}}
\def\be{{\beta}}
\def\ga{{\gamma}}
\def\de{{\delta}}
\def\ep{{\varepsilon}}

\def\si{{\sigma}}
\def\om{{\omega}}
\def\th{{\theta}}

\def\Ga{{\Gamma}}

\def\non{{\nonumber}}

\def\Dc{{\cal D}}

\def\tr{{\rm tr\,}}

\title{{\bf Log-Regularly Varying Scale Mixture of Normals for Robust Regression}}

\date{}
\author{}

\begin{document}

\maketitle
\doublespacing

\vspace{-1.5cm}
\begin{center}
Yasuyuki Hamura$^1$, Kaoru Irie$^2$ and Shonosuke Sugasawa$^3$
\end{center}

\noindent
$^1$Graduate School of Economics, The University of Tokyo\\
$^2$Faculty of Economics, The University of Tokyo\\
$^3$Center for Spatial Information Science, The University of Tokyo

\vspace{5mm}
\begin{center}
{\bf \large Abstract}
\end{center}
Linear regression with the classical normality assumption for the error distribution may lead to an undesirable posterior inference of regression coefficients due to the potential outliers. This paper considers the finite mixture of two components with thin and heavy tails as the error distribution, which has been routinely employed in applied statistics. For the heavily-tailed component, we introduce the novel class of distributions;  their densities are log-regularly varying and have heavier tails than those of Cauchy distribution, yet they are expressed as a scale mixture of normal distributions and enable the efficient posterior inference by Gibbs sampler. We prove the robustness to outliers of the posterior distributions under the proposed models with a minimal set of assumptions, which justifies the use of shrinkage priors with unbounded densities for the coefficient vector in the presence of outliers. The extensive comparison with the existing methods via simulation study shows the improved performance of our model in point and interval estimation, as well as its computational efficiency. Further, we confirm the posterior robustness of our method in the empirical study with the shrinkage priors for regression coefficients.


\newpage 
\section{Introduction}

The robustness to outliers in linear regression models has been well-studied for its importance, and the research on theory and methodology for robust statistics has been accumulated in the past years. Yet, the modeling of error distributions in practice to accommodate outliers has not advanced significantly from Student's $t$-distribution. 
This is contrary to the situations of modern applied statistics where data are enriched by massive observations and the more extreme outliers are expected to be observed and affect the posterior inference. 
Our research aims to contribute to the development of novel error distributions for outlier-robustness which we believe are still in demand.

In the full posterior inference, the concept of robustness is not limited to the point estimation, but targets the whole posterior distributions of parameters of interest. Also known as outlier-proneness or outlier-rejection, the posterior robustness defines the property of posterior distributions that the difference of posteriors with and without outliers diminishes as the values of outliers become extreme \citep{o1979outlier}. The series of research on posterior robustness has revealed several variations of the (sufficient) conditions for error distributions to achieve the robustness, and provided the specific error distributions that meet such conditions; see the detailed review by \cite{o2012bayesian}. 
The recent studies introduce the concept of regularly varying density functions \citep{andrade2006bayesian,andrade2011bayesian}, which are later extended to log-regularly varying functions \citep{desgagne2015robustness,desgagne2019bayesian}, and provide the robustness conditions for the partial and whole posteriors of interest to be unaffected by outliers. As an error distribution whose density function is log-regularly varying, \cite{Gag2019} proposes log-Pareto truncated normal (LPTN) distribution, which replaces the thin-tails of normal distribution by those of heavily-tailed log-Pareto distribution. Despite its desirable property of robustness, the posterior inference for the regression model with the LPTN error distribution is challenging. The class of LPTN distributions has hyperparameters that are difficult to tune and/or estimate, such as the truncation points of Gaussian tails. 
In addition, several parameters cannot be sampled from their conditional posteriors directly, and one has to rely on Metropolis-Hastings algorithm. These factors may lead to the increased computational cost under the LPTN models, which also limits the use of the LPTN distribution under more general linear models including random effects.

We, in contrast, explore a different class of error distributions that have received less attention in the methodological literature. Following \cite{box1968bayesian}, we model the error distribution by the finite mixture of two components; one has thinner tails such as normal distributions, the other is extremely heavily-tailed to accommodate potential outliers, and both are centered at zero. While remaining in the general class of scale mixture of normals \citep{west1984outlier}, this simple, intuitive approach to the modeling of outliers contrasts the literature listed above, where the error is modeled by a single, continuous distribution. 
The structure of finite mixture helps controlling the effect of outliers on the posteriors of parameters of interest, while allowing the conditional conjugacy for posterior computation. For these practical utilities, the finite mixture models have been routinely practiced in applied statistics (see, for example, \citealt{carter1994gibbs}, \citealt{West1997a}, \citealt{fruhwirth2006finite} \citealt{tak2019robust}, and \citealt{da2020bayesian}). 
In this research, we specifically focus on this class of error distributions in proving the posterior robustness.

For the heavily-tailed distribution that comprises the finite mixture, Student's $t$-distribution is still regarded thin-tailed for its outlier sensitivity. We propose the use of distributions that has been utilized in the robust inference for high-dimensional count data \citep{hamura2019global} for their extremely-heavy tails. This is another scale mixture of normals by the gamma distribution with the hierarchical structure on shape parameters, which enables the posterior inference by a simple but efficient Gibbs sampler. The tails of these distributions are heavier than those of Cauchy distributions. In fact, the density of the proposed error distribution is log-regularly varying, as those of other heavily-tailed distributions considered for posterior robustness, including LPTN distributions. 

The proposed finite mixture of the thinly-tailed and heavily-tailed distributions is named the extremely heavily-tailed error (EH) distribution. We prove the posterior robustness under the linear regression models with the EH distribution. The density tails of the EH distribution play an important role in the proof of tail robustness; in fact, the class of error distributions whose density tails are thinner than those of the EH distribution is unable to attain the posterior robustness. The EH distribution is too heavily tailed to have finite moments, but the posterior means and variances of parameters of interest do exist in most situations. 

The set of assumptions required for the proof of posterior robustness is minimal. The assumptions restrict the available priors for the regression coefficients and observational scale, but do not exclude the use of the unbounded prior densities. The posterior robustness is valid even for advanced shrinkage priors, e.g., horseshoe priors \citep{carvalho2009handling,carvalho2010horseshoe}. As a result, the robustness under shrinkage/variable selection is also in the scope of our research. In the empirical studies, we practice the posterior inference for the linear regression models with both the horseshoe prior and the EH distribution for illustration.

The rest of the paper is organized as follows.
In Section \ref{sec:method}, we introduce the new error distribution and describe its use in linear regression models, followed by the theoretical results on the posterior robustness. 
The algorithm for posterior computation is provided in Section~\ref{sec:comp} with the discussion on its computational efficiency. 
In Section~\ref{sec:sim}, we carry out simulation studies to compare the proposed method with existing models, including $t$-distribution and the finite mixture of normal and $t$-distributions. 
In Section~\ref{sec:app}, we illustrate the proposed method using two famous datasets: Boston housing data and diabetes data. 
The paper is concluded with further discussions in Section~\ref{sec:disc}.
The R code implementing the proposed method is available at GitHub repository (\url{https://github.com/sshonosuke/EHE}).

\section{A new error distribution for robust regression}\label{sec:method}

\subsection{Extremely heavy-tailed error distributions}\label{sec:EH}
Let $y_i$ be a response variable and $x_i$ be an associated $p$-dimensional vector of covariates, for $i=1,\ldots,n$.
We consider a linear regression model, $y_i=x_i^t\beta+\sigma\ep_i$, where $\beta$ is a $p$-dimensional vector of regression coefficients and $\sigma$ is an unknown scale parameter. 
The error terms, $\ep_1,\dots ,\ep_n$, are directly linked to the posterior robustness; modeling those errors simply by Gaussian distributions makes the posterior inference very sensitive to outliers.

To achieve the posterior robustness, we introduce a local scale variable $u_i$ and assume that the error distribution is conditionally Gaussian, as $\ep_i|u_i\sim N(0,u_i)$. 
Under this setting, an outlier is explained solely by the extreme value of the error term generated by the higher value of the local scale variable. 
A typical choice of the distribution of $u_i$ is the inverse-gamma distribution, which leads to the marginal distribution of $\ep_i$ being the $t$-distribution. 
However, as shown in \cite{Gag2019} and our main theorem, this choice does not hold the desirable robustness properties of the posterior distribution, even when the distribution of $\ep_i$ is Cauchy distribution. 

As stated in the introduction, the error distribution in this study is not a single continuous mixture of normals, but the mixture of two components. We introduce latent binary variable $z_i$ and model it by ${\rm Pr}[z_i=1] = 1 - {\rm Pr}[z_i=0] = s$ with mixing probability $s\in (0,1)$. If $z_i=0$, then the error distribution is simply the standard normal distribution, i.e., $\epsilon _i | (u_i,z_i=0) \sim N(0,1)$. If $z_i=1$, then we consider the scale mixture of normals with latent scale $u_i$ as $\ep_i|(u_i,z_i=1)\sim N(0,u_i)$. 
The latent scale follows the newly-introduced, extremely heavily-tailed distribution, $u_i \sim H(\cdot ;\gamma)$, where $H$ is the proper probability distribution on $(0,\infty)$ with parameter $\ga > 0$. The density function of $H$-distribution is given by
\begin{equation}\label{U-dist}
H(u;\gamma)={\gamma \over 1 + u} {1 \over \{ 1 + \log (1 + u ) \} ^{1 + \gamma }}, \ \ \ u>0.
\end{equation}
Preparing two distributions in modeling of the error distribution is based on the same modeling philosophy of \cite{box1968bayesian}; the first component generates non-outlying errors and the second component is supposed to absorb outlying errors.
As the model for the variance of outlying errors, the second component $H(\cdot;\gamma)$ is extremely heavily-tailed since $H(u;\gamma)\approx u^{-1}(\log u)^{-1-\gamma}$ as $u\to \infty$, which is known as log-regularly varying density \citep{desgagne2015robustness}. This property is inherited to the marginal distribution of error term $\ep_i$ and plays an important role in the robustness properties of the posterior distribution.

Under the formulation (\ref{U-dist}), the marginal distribution of $\ep_i$ is obtained as 
\begin{equation}\label{EH}
f_{\rm{EH}} ( \ep _i ) = (1 - s) \phi ( \ep _i ; 0, 1) + s \int_{0}^{\infty } \phi ( \ep _i ; 0, u_i) H(u_i; \ga ) du_i \text{,}
\end{equation}
where $\phi (\ep_i;0,u)$ is the normal density with mean zero and variance $u$. 
The second component is the scale mixture of normals, but does not admit any closed-form expression. 
To handle with this component in posterior computation, as we see later in Section~\ref{sec:post}, we utilize the augmentation of $H$-distribution by a couple of gamma-distributed state variables. By this augmentation, the posterior inference for this model is straightforward.

A notable property of the new error distribution is its extremely heavy tails shown in the following proposition, with the proof left in the Appendix. 

\begin{prp}\label{prop:dens}
The density (\ref{EH}) satisfies
\begin{equation*}
f_{\rm EH}(x)\approx |x|^{-1}(\log|x|)^{-1-\gamma}
\end{equation*} 
for large $|x|$ if $s>0$.
\end{prp}

The above proposition shows that the EH distribution directly inherits the heavy tails of the mixing $H$-distribution in the second component of the density in (\ref{EH}). 
As a result, the density of the EH distribution is a family of log-regularly varying functions. In addition, the tails of the EH density are heavier than those of Cauchy distribution; $f_C(x)\approx |x|^{-2}$. 
Based on this observation, we name the new error distribution in (\ref{EH}) {\it extremely heavily-tailed error (EH) distribution}.

The density function in (\ref{EH}) is plotted in Figure~\ref{fig:dens} for $s=0.05, 0.1$ and $0.2$. 
It is observed that the shape of the EH distribution is very similar to one of the standard normal distribution around the origin, whereas the tails become heavier as the mixture weight $s$ increases.
Figure~\ref{fig:dist} shows the cumulative distribution functions (CDFs) of $H$-distributions and the EH distributions. 
The tails of the proposed EH distributions are heavier than those of Cauchy distribution, as seen in the right panel. This fact is also confirmed via the comparison of CDFs of $H$- and inverse-gamma distributions in the left panel. 
It is the property of the EH density shown in these figures that leads to the robustness properties for the posterior distribution, which we show in Theorem~\ref{thm:normalized}.

\begin{figure}[!htb]
\centering
\includegraphics[width=14cm,clip]{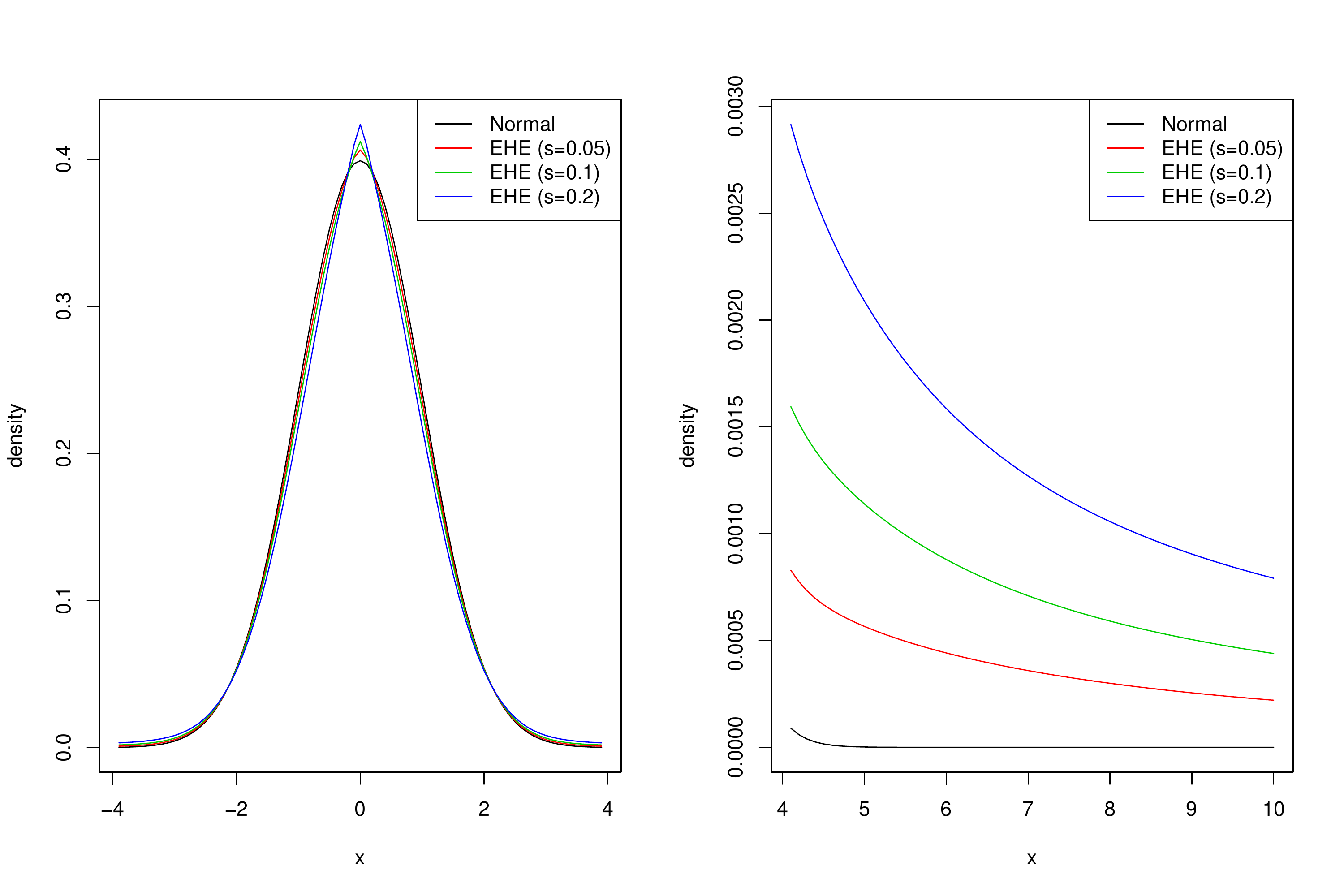}
\caption{Densities of the proposed error distribution with $\gamma=1$ and $s\in \{0.05, 0.1, 0.2\}$ and the standard normal error distribution. 
The intractable integral of the second component is computed by the Monte Carlo integration. 
\label{fig:dens}
}
\end{figure}

\begin{figure}[!htb]
\centering
\includegraphics[width=14cm,clip]{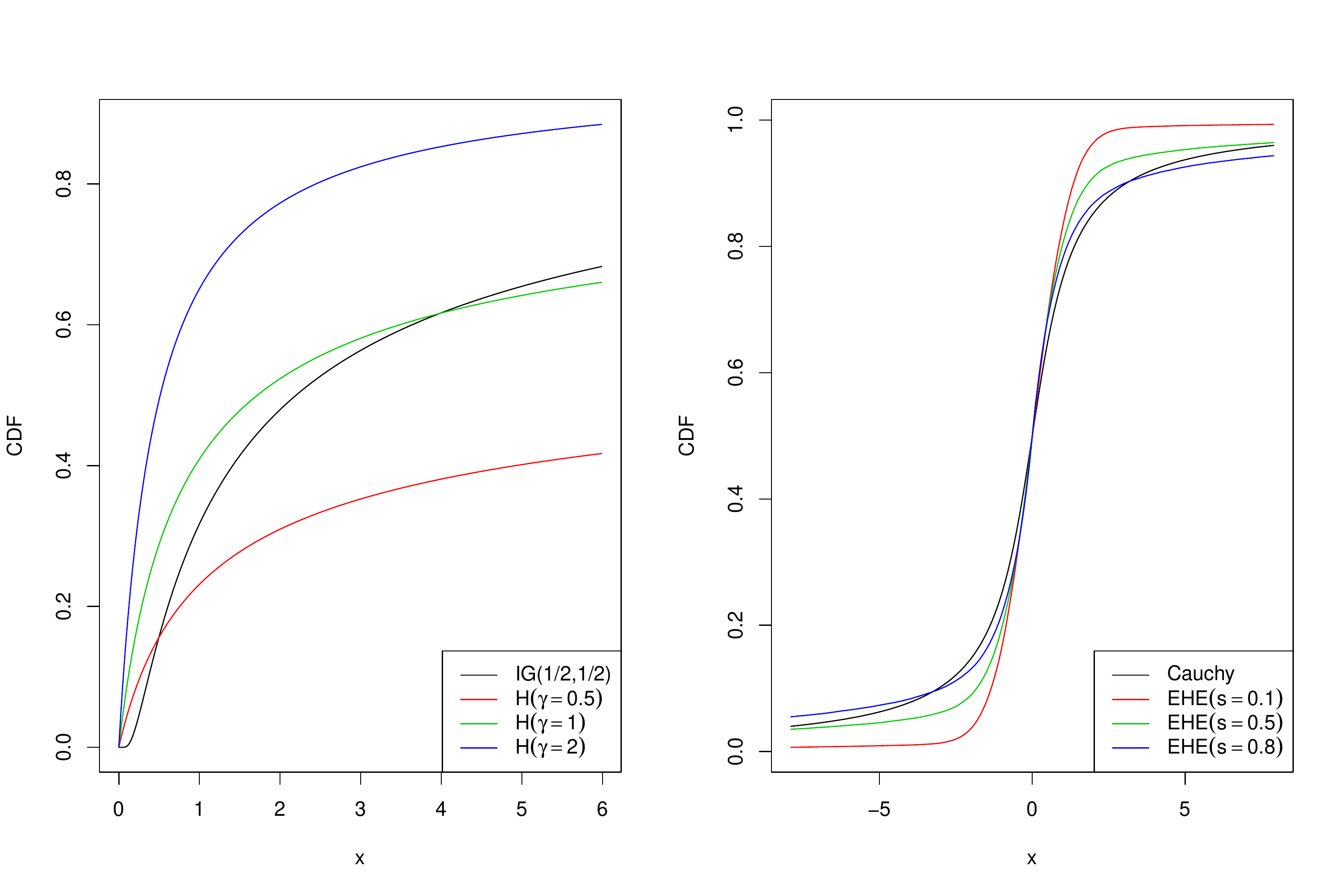}
\caption{Left: Cumulative distribution functions of scale distributions, $H(u;\ga )$ for $\ga \in \{  0.5, 1.0, 2.0 \}$, and the inverse gamma distribution with shape and scale $0.5$. Right: The empirical cumulative distributions of the EH distributions with $\ga = 1$ and $s=0.1,0.5,0.8$ computed by the Monte Carlo integration, compared with the distribution function of Cauchy distribution. 
\label{fig:dist}
}
\end{figure}

\subsection{Definition of outliers}\label{sec:outlier}

We first specify the structure of outliers. Our definition is based on \cite{desgagne2019bayesian}. 
The set of indices for $n$ observations, $\{ 1,\dots ,n \}$, is split into the two disjoint subsets, $\mathcal{K}$ and $\mathcal{L}$, which represent those of the non-outlying and outlying values, respectively. Note that $\mathcal{K} \cup \mathcal{L} = \{ 1,\dots ,n\}$ and $\mathcal{K}\cap \mathcal{L}=\emptyset$. Let $\Dc = \{ y_1,\dots , y_n \}$ be the set of the observed data. The set of the non-outlying observations is defined by $\Dc^{\ast} = \{ y_i | i\in \mathcal{K} \}$.

The concept of (non-)outliers is defined by the observed values specified as, 
\begin{align}
y_i &= \begin{cases} \displaystyle a_i \text{,} & \text{if $i \in \mathcal{K}$} \text{,} \\ \displaystyle a_i + b_i \om \text{,} & \text{if $i \in \mathcal{L}$} \text{,} \end{cases} \non 
\end{align}
where $a_i \in \mathbb{R}$, $b_i \neq 0$ and $\om>0$. We assume that $\om$ is sufficiently large, so that the value of $y_i$ for $i\in \mathcal{L}$ becomes extremely large, either positively or negatively. 
We define the posterior robustness as the limiting behaviors of the posteriors distributions of parameters of interest, $(\beta ,\sigma ^2)$, when $\om$ tends to infinity. 
That is, the model is posterior robust if the two posteriors, one of which is conditioned by the full dataset $\Dc$ and the other of which is conditioned by the dataset without the outliers $\Dc^{\ast}$, are equivalent when $\om \to \infty$. 
To put it in another way, under the posterior robustness, the outlying values are automatically discarded in posterior inference without the knowledge on which observations are outlying.

\subsection{Robustness for the EH prior}\label{sec:robust}

The class of prior distributions for $(\beta ,\sigma ^2)$ for which we prove the posterior robustness is, for $k = 1, \dots , p$, 
\begin{align} \label{prior}
\be _k|\sigma \sim {1 \over \si } \pi_{\beta} \Big( {\be _k \over \si } \Big) \quad \text{and} \quad \si \sim \pi _{\si } ( \si ) \text{,} 
\end{align}
where $\beta_1,\ldots,\beta_p$ are conditionally independent given $\sigma$ and $\pi_{\beta}$ and $\pi _{\sigma}$ are the probability density functions on $\mathbb{R}$ and $(0,\infty)$, respectively.
Let $p( \be , \si | \Dc)$ be the posterior distribution of $(\be,\si)$ under the linear regression model with the EH distribution.  
Under this prior, the following theorem gives sufficient conditions for the posterior with the outliers converges to that without the outliers as $\om \to \infty $. The proof is left in the Supplementary Materials.

\begin{thm}
\label{thm:normalized} 
Assume that there exists $c>0$ such that, 
\begin{itemize}
    \item[(A.1)] $| \mathcal{K} | \ge | \mathcal{L} | + p$, i.e., \#non-outliers $\ge$ \#outliers $+$ \#predictors.
    \item[(A.2)] $\sup_{t \in \mathbb{R}} \{ | t |^c \pi_{\beta} ( t ) \} < \infty $
    \item[(A.3)] The prior moments of $\sigma ^{-|\mathcal{K}|}$, $\sigma ^{c - 1}$ and $\sigma ^{c-n}$ are all finite.
\end{itemize}
Then the linear regression model with the error distribution in (\ref{EH}) and the prior in (\ref{prior}) is posterior robust, i.e., 
\begin{align}
\lim_{\om \to \infty } p( \be , \si | \Dc ) &= p( \be , \si | \Dc^{\ast} ) \non 
\end{align}
for all $( \be , \si ) \in \mathbb{R} ^p \times (0, \infty )$. 
\end{thm}

The three assumptions are met in many examples we encounter in practice. Assumption (A.1) is the requirement for the number of non-outlying observations to be sufficiently large. 
Similar assumptions can be found in the literature (e.g., Theorem~2.1 (ii), \citealt{Gag2019}), but (A.1) is of the simpler form and less restrictive. 
In many situations, the number of the non-outlying observations comprises the majority of the dataset, so that Assumption~(A.1) is satisfied.

Assumption (A.2) limits the choice of priors for $\beta$, but still covers the wide class of probability distributions. For example, this assumption is always satisfied when $\pi_{\beta}(t)$ is bounded and $O(1 / | t |)$ as $| t | \to \infty $. The examples of such prior include the normal and $t$-distributions. Note that, however, (A.2) does not force the prior density $\pi_{\beta}$ to be bounded, unlike the settings of \cite{Gag2019}. As an important example, the horseshoe prior, whose density is unbounded at the origin (Theorem~1, \citealt{carvalho2010horseshoe}), satisfies (A.2) for any $c \in (0, 2]$. As evident in the example of the horseshoe prior, Theorem~\ref{thm:normalized} can be a useful device to check the posterior robustness for the boarder and important class of statistical problems, including the variable selection by the shrinkage priors.

Assumption (A.3) is the moment conditions for observational scale $\sigma$. When the sample size $n$ is large enough and $c\le 1$, then (A.3) is summarized as the existence of negative moments of $\sigma$. In this case, the inverse-gamma distribution for $\sigma ^2$, which is a typical choice of priors in many applications, satisfies (A.3).

\subsection{Tail heaviness for robustness} \label{sec:tail}

Theorem~\ref{thm:normalized} proves the posterior robustness for the linear regression models with the EH distributions, whose density tails are evaluated as $f_{\rm EH}(x) \approx |x|^{-1} (\log |x|)^{-1-\gamma}$, as shown in Proposition~\ref{prop:dens}. These extremely heavy tails are, in fact, the necessary condition for the posterior robustness. To clarify the relationship between the posterior robustness and the tail behavior of the error distributions, we study a wider class of error distributions which includes the proposed distribution as a special case, defined by replacing $H(u;\gamma)$ in (\ref{EH}) with  
\begin{equation}\label{EH2}
H(u;\gamma, \delta)= C(\delta, \ga ) {1\over (1 + u)^{1 + \delta}} {1 \over \{ 1 + \log (1 + u) \} ^{1 + \ga }} \text{,} \quad u > 0 \text{,} \
\end{equation}
where $C(\delta, \ga )$ is a normalizing constant, and $\delta\geq 0$ is an additional shape parameter. Like the degree of freedom of $t$-distributions, the shape parameter $\delta$ is related to the decay of the density tail of (\ref{EH2}), that is, $H(u;\gamma, b)\approx u^{-\delta-1}(\log u)^{-1-\gamma}$. Thus, this class of distributions covers the error distributions whose density tails are lighter than those of the proposed EH distribution in (\ref{EH}), and includes the EH distribution as one with the heaviest tails under $\delta=0$. Note that the density tails become heavier than those of Cauchy distribution if $\delta < 1$. 

It is shown that the choice of hyperparameter that can achieve the posterior robustness is $\delta = 0$ (and arbitrary $\ga > 0$), i.e., the model considered in Theorem~\ref{thm:normalized}. From this observation, we conclude that the tails of the error distribution that are heavier than those of Cauchy distributions is essential for posterior robustness. For details, see the Supplementary Materials.

\subsection{Existence of posterior moments} \label{sec:moment}

The EH distribution is too heavily tailed to have finite moments. However, the posterior of $(\beta ,\sigma ^2)$ has finite means and variances in most situations. We verify this result for the inverse-gamma prior for $\sigma ^2$. 

\begin{prp}\label{prp:moment} Consider the linear regression model with the EH distribution in (\ref{EH}) and the prior for $(\beta ,\sigma)$ given in (\ref{prior}). Furthermore, suppose that the prior for $\sigma ^2$ is an inverse-gamma distribution. 

(a) If (A.2) holds for some $c>0$ and $c \le n$, then $E[ | \be _k |^c | \Dc ] < \infty $ for $k=1,\dots , p$.

(b) If $d\le n$, then $E[ \sigma ^d | \Dc ] < \infty$.
\end{prp}
It is immediate from (a) that the posterior means and variances of coefficients $\beta$ exist under the horseshoe prior for $\beta$, which is given later in (\ref{shrink}). 

\begin{cor}
If the prior for $\beta$ is horseshoe and $n \ge 2$, then $E[ |\be _k|^2 | y ] < \infty $. 
\end{cor}

The proof is given in the Supplementary Materials. In fact, the existence of posterior moments of $(\beta ,\sigma^2)$ can be discussed for the broad class of error distributions and priors for $(\beta , \sigma)$, not being limited to the linear regression model we particularly consider in this paper. Proposition~\ref{prp:moment} is proved with such generality.

\section{Posterior Computation}\label{sec:comp}

\subsection{Gibbs sampler by augmentation}\label{sec:post}
An important property of the proposed EH distribution (\ref{EH}) is its computational tractability, that is, we can easily construct a simple Gibbs sampling for posterior inference. 
Note that the error distribution contains two unknown parameters, $s$ and $\gamma$, and we can adopt conditionally conjugate priors given by $s\sim {\rm Beta}(a_s, b_s)$ and $\gamma\sim {\rm Ga}(a_{\gamma}, b_{\gamma})$.
The conditionally conjugate priors can also be found for main parameters, $\beta$ and $\sigma^2$, and we use $\beta \sim N(A_{\beta}, B_{\beta})$ and $\sigma^{-2}\sim {\rm Ga}(a_{\sigma}, b_{\sigma})$. 
The multivariate normal prior for $\beta$ can be replaced with the scale mixture of normals, such as shrinkage priors, which is discussed later in Section~\ref{sec:shrink}.

To derive the tractable conditional posteriors, we need to keep the likelihood conditionally Gaussian with latent scale $u_i$. This can be done easily by conditioning the set of latent variables $(z_i,u_i)$. 
The conditional conjugacy for $(\beta , \sigma ^2)$ follows immediately from the conditionally Gaussian likelihoods.

The full conditional distributions of the other parameters and latent variables in the EH distribution are not any well-known distribution. However, we can augment the model with latent parameters by utilizing the following integral expression of density $H(u_{i};\ga)$,
$$
H(u_{i};\gamma)=\iint_{(0,\infty)^2} {\rm Ga}(u_{i};1,v_i){\rm Ga}(v_i; w_i, 1){\rm Ga}(w_i;\gamma, 1)dv_idw_i.
$$ 
Namely, the random variable $u_{i}$ following the density $H(u_{i};\gamma)$ admits the mixture representation: $u_{i}|(v_i,w_i)\sim {\rm Ga}(1, v_i)$, $v_i|w_i\sim {\rm Ga}(w_i, 1)$ and $w_i\sim {\rm Ga}(\gamma,1)$, which enables us to easily generate samples from the full conditional distribution of $(u_{i}|v_i,w_i)$ and $(v_i,w_i|u_{i})$.

The introduction of the two latent states, $(v_i,w_i)$, is useful in deriving the conditional posterior of $u_i$, and the algorithm of Gibbs sampler immediately follows with latent $(v_i,w_i)$ as the part of the Markov chain, although $(v_i,w_i)$ is totally redundant in posterior sampling of the other parameters. We marginalize $(v_i,w_i)$ out when sampling $\gamma$, $s$ and $z_i$'s from their conditional posteriors. This modification of the original Gibbs sampler simplifies the sampling procedure, and even facilitates the mixing, while targeting the same stationary distribution of the original Markov chain (Partially collapsed Gibbs sampler, \citealt{van2008partially}). The algorithm for posterior sampling is summarized as follows.

\vspace{0.5cm}
\noindent
{\bf Summary of the posterior sampling }
\begin{itemize}
\item[-]
Sample $\beta$ from the full conditional distribution $N(\Bt\At, \Bt)$, where 
\begin{align*}
&\Bt^{-1}=B_{\beta}^{-1}+\sigma^{-2}X^tDX, \ \ \ \  \At=B_{\beta}^{-1}A_{\beta}+\sigma^{-2}X^tDY
\end{align*} 
with $D={\rm diag}(u_1^{-z_1},\ldots,u_n^{-z_n})$.

\item[-]
Sample $\sigma^{-2}$ from ${\rm Ga}(\at_\sigma, \bt_\sigma)$, where 
$$
\at_\sigma=a_{\sigma}+n/2,  \ \ \ \ \bt_\sigma=b_{\sigma}+\sum_{i=1}^n(y_i-x_i^t\beta)^2/2u_i^{z_i}
$$

\item[-]
Sample $z_i$ from Bernoulli distribution; the probabilities of $z_i=0$ and $z_i=1$ are proportional to $(1-s)\phi(y_i;x_i^t\beta, \sigma^2)$ and $s\phi(y_i;x_i^t\beta, \sigma^2 u_i)$, respectively. 

\item[-]
The full conditional distributions of $s$ and $\gamma$ are given by ${\rm Beta}(\at_s, \bt_s)$ and ${\rm Ga}(\at_{\gamma}, \bt_{\gamma})$, respectively, where $\at_s=a_s+\sum_{i=1}^nz_i$ and $\bt_s=b_s+n-\sum_{i=1}^nz_i$, $\at_\gamma=a_{\gamma}+n$ and $\bt_\gamma=b_{\gamma}+\sum_{i=1}^n\log\{1+\log(1+u_{i})\}$.

\item[-] 
For each $i$, independently, sample $(v_i,w_i)$ first in a compositional way; sample $w_i$ from ${\rm Ga}(1+\gamma, 1+\log(1+u_{i}))$ and  $(v_i|w_i)$ as ${\rm Ga}(1+w_i, 1+u_{i})$. Then, sample $u_{i}$ from ${\rm GIG}(1/2, 2v_i, (y_i-x_i^t\beta)^2/\sigma^2)$ if $z_i=1$ or from ${\rm Ga}(1, v_i)$ if $z_i=0$.

\end{itemize}

We finally remark the choice of hyperparameters in the priors for $s$ and $\gamma$. 
Despite the EH distribution is log-regularly varying under arbitrary $\gamma>0$, the use of a large value of $\gamma$ is not suitable to capture potential outliers since the tail of EH gets lighter as $\gamma$ increases.
Moreover, the use of different values of $\gamma$ would not considerably affect the posterior result as long as $\gamma$ is not large. 
Hence, instead of using a diffuse prior for $\gamma$, we rather recommend simply using a fixed value.
In particular, we adopt $\gamma=1$ as the default choice, and its sensitivity will be investigated in Section \ref{sec:sim}.
As a more data-dependent way, we also recommend employing an informative prior that prevent large values of $\gamma$ by setting, for example,  $a_{\gamma}=b_{\gamma}=100$, which will be considered in Section \ref{sec:sim}
Regarding the mixing proportion $s$, we adopt $a_s=b_s=1$ resulting the uniform prior for $s$ as a default choice.

\subsection{Efficiency in computation}
A possible reason that the finite mixture has attracted less attention in the past research on posterior robustness is, as mentioned in \cite{desgagne2019bayesian}, the increased number of latent state variables introduced by augmentation, and the concern for the efficiency of posterior computation. It is the same concern seen in Bayesian variable selection \citep{george1993variable}; the finite mixture model for the prior on regression coefficients results in the necessity of stochastic search in the high-dimensional model space, hence causes the slow convergence of Markov chains and the costly computation. 
It is clear in the above algorithm, however, that the use of finite mixture as error distributions is completely different from the variable selection in terms of the model structure and free from such computational problem. 
Unlike the variable selection, the membership of each $i$ to either of the two components in our model is independent of one another, which facilitates the stochastic search in $2^n$ possible combination of the model space. This fact also shows that the sampling of $(z_i,u_i,v_i,w_i)$ can be done completely in parallel across $i$'s, hence our algorithm is scaled and computational feasible for the dataset with extremely large $n$. We continue to discuss the computational efficiency of the finite mixture approach in Section~4 through the extensive comparison with other models by using the simulated dataset.

\subsection{Robust Bayesian variable selection with shrinkage priors}\label{sec:shrink}
When the dimension of $x_i$ is moderate or large, it is desirable to select a suitable subset of $x_i$ to achieve efficient estimation. 
This procedure of variable selection would also be seriously affected by the possible outliers, by which we may fail to select suitable subsets of covariates. 
For a robust Bayesian variable selection procedure, we introduce shrinkage priors for regression coefficients. 
Here we rewrite the regression model to explicitly express an intercept term as $y_i=\alpha+x_i^{t}\beta+\ep_i$, and 
consider a normal prior $\alpha\sim N(0,A_{\alpha})$ with fixed hyperparameter $A_{\alpha}>0$.
For the regression coefficients $\beta$, we consider a class of independent priors expressed as a scale mixture of normals given by 
\begin{equation}\label{shrink}
\pi(\beta)=\prod_{k=1}^p\int_{0}^\infty \phi(\beta_k; 0, \sigma^2\tau^2\xi_k)\pi _{\xi}(\xi_k){\rm d}\xi_k,
\end{equation}
where $\pi _{\xi}(\cdot)$ is a mixing distribution, and $\kappa^2$ is an unknown global parameter that controls the strength of the shrinkage effects. 
Examples of the mixing distribution $\pi _{\xi}(\cdot)$ includes the exponential distribution leading to the Laplace prior of $\beta$ (Bayesian Lasso, \citealt{park2008bayesian}), and the half-Cauchy distribution for $\xi_k^{1/2}$ which results in the horseshoe prior \citep{carvalho2009handling,carvalho2010horseshoe}.
The robustness property of the resulting posterior distributions is guaranteed for those shrinkage priors because Assumption (A.2) of Theorem~\ref{thm:normalized} is satisfied.

In terms of posterior computation, the key property is that the conditional distribution of $\beta_k$ given $\xi_k$ under (\ref{shrink}) is a normal distribution, so the sampling algorithm given in Section~\ref{sec:post} is still valid with minor modification.
Specifically, the sampling from the full conditional distributions of $\alpha$, $\beta$, $\sigma^2$ and $\xi_1,\ldots,\xi_p$ is modified or newly added as follows:

\begin{itemize}
\item[-]
Sample $\alpha$ from $N(\At_{\alpha}^{-1}\Bt_{\alpha}, \At_{\alpha}^{-1})$, where  
\begin{align*}
\At_{\alpha}=A_{\alpha} + \sigma^{-2}\sum_{i=1}^n u_i^{-1}, \ \ \ \ \Bt_{\alpha}=\sigma^{-2}\sum_{i=1}^nu_i^{-1}(y_i-x_i^t\beta).
\end{align*} 

\item[-]
Sample $\beta$ from $N(\At_{\beta}^{-1} X^tD\widetilde{Y}, \sigma^2\At_{\beta}^{-1})$, where
\begin{align*}
\widetilde{Y}=Y-\alpha 1_n, \ \ \ \ \At_{\beta}=\Lambda^{-1}+X^tDX, \ \ \ \mbox{with} \ \ \ \Lambda=\tau^2{\rm diag}(\xi_1,\ldots,\xi_p).
\end{align*}

\item[-]
Sample $\sigma^{-2}$ from ${\rm Ga}(\at_\sigma, \bt_\sigma)$, where 
$$
\at_\sigma=a_{\sigma}+(n+p)/2,  \ \ \ \ \bt_\sigma=b_{\sigma}+\sum_{i=1}^n(y_i-x_i^t\beta)^2/2u_i^{z_i}+\beta^t\Lambda^{-1}\beta.
$$

\item[-]
Sample $\xi_k$ for each $k$ and $\tau^2$ from their full conditionals. Their densities are proportional to $\phi(\beta_k; 0, \sigma^2\tau^2\xi_k)\pi _{\xi}(\xi_k)$ and $\pi_{\tau^2}(\tau^2)\prod_{k=1}^p\phi(\beta_k; 0, \sigma^2\tau^2\xi_k)$, respectively, where $\pi_{\tau^2}(\tau^2)$ is a prior density for $\tau^2$. 

\end{itemize}

The full conditional distributions of $\alpha$ and $\beta$ are familiar forms thanks to the normal mixture representation of the EH distribution and the shrinkage priors. 
The sampling of $\xi_k$ and $\tau^2$ depends on the choice of shrinkage priors, but the existing algorithms in the literature can be directly imported to our method.

In Section~\ref{sec:app}, we adopt the horseshoe prior for regression coefficients with the EH distribution for the error terms. We here provide the details of sampling algorithm under the horseshoe model.
The horseshoe prior assumes that $\sqrt{\xi _k} \sim C^{+}(0,1)$ independently for $k=1,\dots ,p$ and $\tau\sim C^{+}(0,1)$, where $C^{+}(0,1)$ is the standard half-Cauchy distribution with probability density function given by $p(x) =2/\pi(1+x^2)$ for $x>0$.
Note that they admit hierarchical expressions given by $\xi_k|\lambda_k\sim {\rm IG}(1/2, 1/\lambda_k)$ and $\lambda_k\sim {\rm IG}(1/2, 1/2)$ for $\xi_k$, and $\tau^2|\nu\sim {\rm IG}(1/2, 1/\nu)$ and $\nu\sim {\rm IG}(1/2, 1/2)$ for $\tau^2$.
Then, one can sample from each full conditional distribution as follows: 
\begin{itemize}
\item[-]
Sample $\xi_k$ from ${\rm IG}(1, 1/\lambda_k+\beta_k^2/2\tau^2\sigma^2)$.
\item[-]
Sample $\lambda_k$ from ${\rm IG}(1, 1+1/\xi_k)$.
\item[-]
Sample $\tau^2$ from ${\rm IG}((p+1)/2, 1/\nu+\sum_{k=1}^p\beta_k^2/2\xi_k\sigma^2)$.
\item[-]
Sample $\nu$ from ${\rm IG}(1, 1+1/\tau^2)$.
\end{itemize}
These sampling steps can be directly incorporated into the Gibbs sampling algorithm given in Section \ref{sec:post}.

\subsection{Beyond linear regression}\label{sec:extension}
The proposed error distribution can be adopted in more general linear regression models. As an example, we consider a hierarchical model given by 
\begin{equation}\label{general-LM}
y_i=x_i^t\beta+g_i^t b+\ep_i, \ \ \ i=1,\ldots,n,
\end{equation}
where $g_i$ is a $r$-vector of additional covariates and $b$ is a vector of random effects distributed as $b\sim N(0, H(\psi))$ with $r{\times}r$ covariance matrix $H(\psi)$ parametrized by $\psi$. 
To absorb potential effects of outliers, we use the EH distribution for $\ep_i$.
The model structure (\ref{general-LM}) is general enough to represent a wide variety of useful models, as seen in the later sections. 
Even under the model (\ref{general-LM}), the robustness properties for $\beta$ demonstrated in Section~\ref{sec:robust} can be discussed by checking whether the prior for $b$ satisfies Assumption (A.2). 
Moreover, the augmentation strategy for the efficient posterior computation algorithm can still be employed and the full conditional distribution of $b$ is normal.  
We adopt a random intercept model for longitudinal data in our simulation study in Section \ref{sec:sim-RI} 
and a linear regression with spatial effects in our application in Section \ref{sec:boston}.

\section{Simulation studies}\label{sec:sim}

\subsection{Linear regression}
We here carry out simulation studies to investigate the performance of the proposed method together with existing methods.
We generated $n=300$ observations from the linear regression model with $p=20$ covariates, given by 
$$
y_i= \beta_0 +\sum_{k=1}^p \beta_kx_{ik} + \sigma\ep_i, \ \ \ i=1,\ldots,n,
$$
where $\beta_0=0.5, \beta_1=\beta_4=0.3$, $\beta_7=\beta_{10}=2$, $\sigma=0.5$ and the other coefficients are set to $0$.
Here the vector of covariates $(x_{i1},\ldots,x_{ip})$ was generated from a multivariate normal distribution with zero mean vector and variance-covariance matrix whose $(k,\ell)$-entry has $(0.2)^{|k-\ell|}$ for $k,\ell \in \{1,\ldots,p\}$.
Regarding the contamination structure of the error term, we adopted the location-shift model \citep{abraham1978linear};
$$
\ep_i\sim (1-\omega)N(0,1) + \omega N(\mu ,1), \ \ \ i=1,\ldots,n,
$$
where $\omega$ is the contamination ratio and $\mu$ is the location of outliers. 
We considered all the combinations of $\omega\in \{ 0.05, 0.1\}$ and $\mu \in\{5, 10, 15, 20\}$, in addition to the case of no contamination ($\om = 0$), which leads to 9 scenarios in total. Under this setting, we replicate 500 datasets independently.

For each of the 500 simulation datasets, we applied the following robust regression methods. The error distributions we consider include with the EH distribution, the LPTN distribution \citep{Gag2019}, and $t$-distribution with $\nu$ degrees of freedom.
For the hyperparameter $\ga$ in the EH distribution, we fixed $\gamma=1$ (denoted by EH) and estimated $\gamma$ adaptively (aEH) by assigning ${\rm Ga}(100, 100)$ prior distribution. 
For the LPTN distribution, the tuning parameter $\rho\in (2\Phi(1)-1, 1)\approx (0.6827, 1)$ is specified as $\rho=0.9$ and $\rho=0.7$ (LP1 and LP2, respectively).
Regarding the degree of freedom $\nu$ in the $t$-distribution, we specifically selected the results of $\nu=1$ (Cauchy distribution, denoted by C), $\nu=3$ (T3), and an adaptive version (aT) that employs a discrete uniform prior on $\nu\in\{1,2,3,4,5,8,10,15,20,30,50\}$. 
In addition, the two-component mixture of the $t$-distribution with $\nu=1/2$ and the standard normal distributions is considered (MT). 
We also employed the EH distribution with $\gamma=0.5$ and $\gamma=0.2$ to assess its sensitivity, $t$-distribution with $\nu=2.1$, and the MT distribution with $\nu=2.1$.
To save the space, we reported the results of these four methods 
in the Supplementary Material. 
As a standard method, we adopted the normal distribution as the error distribution (denoted by N) that should perform best in the absence of outliers. Note that all the error distributions listed here are ``misspecified'' for missing the location shift of the error term in the data generating process. This setting emphasizes that the posterior robustness verified in this research is valid regardless of the structure of outliers.

The priors for the regression coefficients and observational scale are set as $\beta_k\sim N(0, 1000)$ and $\sigma^{-2}\sim {\rm Ga}(1, 1)$ for all the models. 
To employ the posterior inference, we generated the posterior samples of $(\beta ,\sigma)$ by Gibbs sampler under the EH, $t$ and normal error distributions. For the LPTN distribution, the random-walk Metropolis-Hastings algorithm was adopted as in \cite{Gag2019}, in which the step sizes were set to $0.05$. 
For each of the 9 models, we generated 3000 posterior samples after discarding the first 1000 samples.

Based on the posterior samples, we computed posterior means as well as $95\%$ credible intervals of $\beta_k$ for $k=1,\dots ,p$. 
The performance of the point and interval estimation was assessed by square root of mean squared errors (RMSE), coverage probabilities (CP) and average length (AL) based on the 500 replications of the simulation, and these values were averaged over $\beta_0, \ldots, \beta_p$. 
In addition, we evaluated the efficiency of the sampling schemes by computing the average of inefficient factors (IF) of the posterior samples. 

In Table \ref{tab:sim}, we reported the values of these performance measures in 9 scenarios. 
When $\omega=0$ (no outlier), as easily predicted, the normal error distribution provides the most efficient result in all measures. While the other methods are slightly inefficient, the proposed method (EH and aEH in the table) performs almost in the same way as the normal distribution. This is an empirical evidence that the efficiency loss of the EH distribution is very limited owing to the normal component in the mixture. 
In the other robust methods, MSEs are slightly higher than the that of the normal distribution and CPs are smaller than the nominal level. 

In the other scenarios, where outliers are incorporated in the data generating process, the performance of the normal distribution is significantly lowered, and the robustness property is highlighted in the performance measures of the other models. 
In particular, the EH distribution with fixed $\ga$ (EH) performs quite stably in both point and interval estimation. 
The adaptive version (aEH) also works reasonably well, and the performance is comparable with EH. 
The LPTN model with $\rho = 0.9$ (LP1) shows reasonable performance in point estimation, but its CPs tend to be smaller than the nominal level. The other LPTN model with $\rho = 0.7$ (LP2) greatly worsens the accuracy of point estimation, implying the sensitivity of the choice of hyperparameter $\rho$ to the posteriors. 
The other models (C, T3, aT and MT) also suffer from the larger MSE values, especially in the scenarios of large $\om$ and $\mu$, which emphasizes the lack of posterior robustness under the $t$-distribution family. 
In addition, the interval estimation under the $t$-distributions depends on the degree-of-freedom parameter, as seen in the results of Cauchy and $t_3$-distributions where the credible intervals are too wide and narrow, respectively.

In terms of computational efficiency, it is remarkable that the IF values of the EH models are small and comparable with those of the $t$-distribution methods, which shows the efficiency of the proposed Gibbs sampling algorithm. 
On the other hand, the IFs of the LPTN models are very large due to the use of Metropolis-Hastings algorithm. 
To obtain the reliable posterior analysis under the LPTN models, one needs to increase the number of iterations in the computation by MCMC, or to spend more effort tuning the step-size parameter. We observed that the performance of LPTNs is improved under the simpler settings of less covariates ($p=10$), but the overall result of model comparison remains almost the same. 
See the Supplementary Materials for this additional experiment. 
Moreover, we measured the actual computation time of five methods (EH, LP1, T3, MT and N) under larger sample sizes, which are reported in the Supplementary Materials.

Finally, we evaluate the predictive performance.
We generated $m=20$ additional covariates $x_{j\ast} \ (j=1,\ldots,m)$ from the same multivariate normal distribution, and then generated true response value $y_{j\ast}$ based on the linear regression with $\ep_i\sim N(0,1)$. That is, the predicted response is not contaminated with outliers. 
Accordingly, in prediction with the EH and MT distributions, we construct the sampling model of $y_{j\ast}$ conditional on $z_j=0$ as  
$$
f(y_{j\ast}|\Dc,z_j=0)=\int \phi(y_{j\ast};x_{j\ast}^t\beta, \sigma^2)\pi(\beta,\sigma|\Dc)d\beta d\sigma. 
$$
This predictive distribution reflects our belief that the prediction should be considered only for non-outlying observations. If one believes that the predicted response might also be outlying, then the model in (\ref{EH}) can be used for prediction, being unconditional on $z_j$, at the cost of inflated predictive uncertainty. To handle with the outlying predictive values, however, the models for outlier detection should be more appropriate (e.g., \citealt{desgagne2021efficient}). 
For the LPTN and $t$-distributions, it is difficult to separate non-outliers and outliers in the model. For these models, we used the same error distributions for prediction. We reported the result of T3 model only; the $95\%$ predictive intervals of $y_{j\ast}$ under the LPTN and other $t$ models are extremely wide due to their heavy tails.

To evaluate the predictive performance, we computed MSE of the posterior predictive mean and CP and AL of $95\%$ predictive intervals of $y_{j\ast}$. 
These values were averaged over 500 replications, which are shown in Table \ref{tab:sim-pred}.
First, it can be seen that the model with the Gaussian errors produces worse point predictions and wider interval estimates as more and larger outliers are generated, which is clearly due to the lack of posterior robustness. 
The other robust methods are equally performative in terms of point prediction, but they show a great difference in the uncertainty quantification. 
The T3 method tends to be too conservative, in the sense that the predictive intervals are too wide and show the almost 100\% coverage. 
The EH and MT models have the similar predictive results, while the coverage rates suggest the potential under-coverage of the MT model. 
This result shows the importance of posterior robustness, or the use of error distributions with extremely heavy tails in estimation, for not only posterior inference but also predictive analysis.

\begin{table}[htbp]
\caption{Average values of RMSEs, CPs, ALs and IFs of the proposed extremely-heavy tailed distribution with $\gamma$ fixed (EH) and estimated (aEH), log-Pareto normal distribution with $\rho=0.9$ (LP1) and $\rho=0.7$ (LP2), Cauchy distribution (C), $t$-distribution with 3 degrees of freedom (T3) and estimated degrees of freedom (aT), two component mixture of normal and $t$-distribution with $1/2$ degrees of freedom (MT), and normal linear regression (N), based on 500 replications in 9 combinations of $(100\omega, \mu)$.
All values except for IFs are multiplied by 100.
\label{tab:sim}}
\begin{center}
{\small
\begin{tabular}{ccccccccccccc}
\hline
 &  &  & EH & aEH & LP1 & LP2 & C & T3 & aT & MT & N \\
 \hline
 & (0, --) &  & 6.32 & 6.32 & 6.71 & 7.95 & 7.84 & 6.78 & 6.57 & 6.32 & 6.32 \\
 & (5, 5) &  & 6.90 & 6.93 & 7.11 & 8.31 & 7.95 & 7.13 & 7.38 & 6.82 & 10.67 \\
 & (10, 5) &  & 8.92 & 8.52 & 8.96 & 9.29 & 8.35 & 8.33 & 9.72 & 9.59 & 15.85 \\
 & (5, 10) &  & 6.61 & 6.64 & 6.88 & 8.22 & 7.83 & 6.88 & 7.19 & 6.57 & 18.82 \\
RMSE & (10, 10) &  & 7.05 & 7.17 & 7.12 & 8.34 & 8.04 & 7.43 & 10.05 & 7.68 & 29.68 \\
 & (5, 15) &  & 6.61 & 6.65 & 6.83 & 8.11 & 7.88 & 6.87 & 7.07 & 6.58 & 27.36 \\
 & (10, 15) &  & 7.07 & 7.15 & 7.26 & 8.29 & 8.04 & 7.16 & 10.00 & 8.01 & 43.36 \\
 & (5, 20) &  & 6.54 & 6.58 & 6.76 & 8.08 & 7.86 & 6.79 & 6.93 & 6.51 & 36.10 \\
 & (10, 20) &  & 7.03 & 7.11 & 7.00 & 8.28 & 7.98 & 7.06 & 10.38 & 8.05 & 58.08 \\
 \hline
 & (0, --) &  & 94.9 & 94.8 & 90.2 & 72.3 & 88.6 & 93.2 & 94.5 & 95.0 & 94.8 \\
 & (5, 5) &  & 94.8 & 94.7 & 91.9 & 77.0 & 89.5 & 94.7 & 95.8 & 94.7 & 91.1 \\
 & (10, 5) &  & 93.2 & 93.2 & 91.8 & 79.8 & 90.7 & 94.2 & 94.6 & 93.2 & 90.5 \\
 & (5, 10) &  & 94.9 & 94.9 & 91.9 & 75.7 & 90.4 & 95.5 & 97.8 & 94.8 & 90.2 \\
CP & (10, 10) &  & 94.4 & 94.1 & 93.1 & 78.8 & 91.3 & 96.9 & 98.0 & 94.5 & 90.3 \\
 & (5, 15) &  & 94.7 & 94.6 & 91.5 & 75.5 & 90.2 & 95.5 & 98.4 & 94.9 & 90.1 \\
 & (10, 15) &  & 94.1 & 93.7 & 92.9 & 78.9 & 91.3 & 97.3 & 99.0 & 94.4 & 90.1 \\
 & (5, 20) &  & 94.9 & 94.7 & 92.0 & 77.2 & 90.8 & 96.1 & 98.9 & 94.9 & 90.4 \\
 & (10, 20) &  & 94.4 & 94.2 & 92.9 & 78.4 & 91.9 & 97.7 & 99.5 & 94.6 & 90.5 \\
 \hline
 & (0, --) &  & 24.7 & 24.7 & 23.2 & 18.6 & 24.7 & 24.7 & 25.1 & 24.7 & 24.7 \\
 & (5, 5) &  & 26.7 & 26.7 & 26.1 & 21.1 & 26.1 & 27.6 & 30.4 & 26.4 & 36.3 \\
 & (10, 5) &  & 30.4 & 29.9 & 31.4 & 24.7 & 28.2 & 32.0 & 37.4 & 30.7 & 44.3 \\
 & (5, 10) &  & 25.9 & 26.0 & 25.2 & 20.2 & 26.2 & 28.0 & 34.3 & 25.8 & 59.1 \\
AL & (10, 10) &  & 27.2 & 27.4 & 27.3 & 22.3 & 27.9 & 32.7 & 49.4 & 27.7 & 77.4 \\
 & (5, 15) &  & 25.7 & 25.7 & 24.9 & 20.2 & 26.1 & 27.9 & 36.2 & 25.7 & 83.9 \\
 & (10, 15) &  & 27.0 & 27.0 & 27.0 & 22.2 & 27.7 & 32.6 & 58.9 & 27.6 & 112.3 \\
 & (5, 20) &  & 25.8 & 25.8 & 24.8 & 20.5 & 26.2 & 28.1 & 37.7 & 25.8 & 110.2 \\
 & (10, 20) &  & 26.9 & 26.9 & 26.5 & 21.7 & 27.9 & 32.8 & 69.2 & 27.4 & 149.3 \\
 \hline
 & (0, --) &  & 1.01 & 1.01 & 44.92 & 54.00 & 4.68 & 2.12 & 1.86 & 0.99 & 0.98 \\
 & (5, 5) &  & 2.09 & 2.41 & 43.03 & 53.17 & 4.32 & 1.97 & 1.80 & 1.36 & 0.99 \\
 & (10, 5) &  & 3.34 & 4.26 & 40.52 & 51.95 & 3.98 & 1.86 & 1.82 & 1.96 & 0.98 \\
 & (5, 10) &  & 1.98 & 2.27 & 43.45 & 53.45 & 4.22 & 1.89 & 1.79 & 1.28 & 0.98 \\
IF & (10, 10) &  & 3.05 & 3.68 & 41.72 & 52.64 & 3.84 & 1.70 & 1.95 & 1.68 & 0.98 \\
 & (5, 15) &  & 1.96 & 2.21 & 43.73 & 53.36 & 4.22 & 1.88 & 1.77 & 1.28 & 0.98 \\
 & (10, 15) &  & 3.05 & 3.54 & 42.26 & 52.74 & 3.84 & 1.67 & 2.04 & 1.64 & 0.98 \\
 & (5, 20) &  & 1.97 & 2.19 & 43.58 & 53.44 & 4.23 & 1.87 & 1.75 & 1.28 & 0.98 \\
 & (10, 20) &  & 3.05 & 3.49 & 42.52 & 52.87 & 3.84 & 1.65 & 2.15 & 1.59 & 0.98 \\
\hline
\end{tabular}
}
\end{center}
\end{table}

\begin{table}[htbp]
\caption{Average values of RMSEs of posterior predictive means and CPs and ALs of $95\%$ prediction intervals based on the EH method with $\gamma=1$, $t$-distribution with 3 degrees of freedom (T3), two component mixture of normal and $t$-distribution with $1/2$ degrees of freedom (MT), and the standard normal linear regression (N), based on 500 replications in 9 combinations of $(100\omega, \mu)$.
RMSE and CP are multiplied by 100.
\label{tab:sim-pred}}
\begin{center}
{\small
\begin{tabular}{ccccccccccccc}
\hline
 & $(100\omega, \mu)$  & EH & T3 & MT & N \\
 \hline
 & (0, --) & 50.5 & 50.7 & 50.5 & 50.5 \\
 & (5, 5) & 50.6 & 50.7 & 50.5 & 53.0 \\
 & (10, 5) & 51.3 & 51.3 & 51.7 & 58.2 \\
 & (5, 10) & 51.6 & 51.7 & 51.5 & 61.2 \\
RMSE & (10, 10) & 50.6 & 50.8 & 50.7 & 78.1 \\
 & (5, 15) & 50.5 & 50.6 & 50.5 & 70.1 \\
 & (10, 15) & 51.2 & 51.3 & 51.2 & 99.0 \\
 & (5, 20) & 51.2 & 51.3 & 51.2 & 85.3 \\
 & (10, 20) & 50.8 & 50.9 & 52.2 & 127.7 \\
 \hline
 & (0, --) & 95.5 & 98.9 & 95.6 & 95.8 \\
 & (5, 5) & 95.7 & 99.6 & 95.0 & 99.2 \\
 & (10, 5) & 96.7 & 99.9 & 95.0 & 99.9 \\
 & (5, 10) & 94.6 & 99.4 & 94.0 & 100.0 \\
CP & (10, 10) & 95.6 & 100.0 & 94.5 & 100.0 \\
 & (5, 15) & 95.2 & 99.7 & 94.7 & 100.0 \\
 & (10, 15) & 95.1 & 100.0 & 93.9 & 100.0 \\
 & (5, 20) & 94.6 & 99.7 & 94.2 & 100.0 \\
 & (10, 20) & 95.2 & 99.9 & 94.2 & 100.0 \\
 \hline
 & (0, --) & 2.01 & 2.64 & 2.02 & 2.03 \\
 & (5, 5) & 2.06 & 2.98 & 2.00 & 2.97 \\
 & (10, 5) & 2.21 & 3.46 & 2.14 & 3.62 \\
 & (5, 10) & 2.01 & 3.02 & 1.96 & 4.80 \\
AL & (10, 10) & 2.04 & 3.64 & 1.95 & 6.35 \\
 & (5, 15) & 2.00 & 3.02 & 1.96 & 6.77 \\
 & (10, 15) & 2.02 & 3.62 & 1.94 & 9.12 \\
 & (5, 20) & 1.99 & 3.04 & 1.95 & 8.96 \\
 & (10, 20) & 2.01 & 3.68 & 2.00 & 12.14 \\
\hline
\end{tabular}
}
\end{center}
\end{table}

\subsection{Random intercept models}\label{sec:sim-RI}
Next, we consider simulation studies using the following random intercept model: 
\begin{equation}\label{RI}
y_{jt}= \beta_0 +\sum_{k=1}^p \beta_kx_{jtk} + v_j + \sigma\ep_{jt}, \ \ \ t=1,\ldots,T, \ \ \ j=1,\ldots,m, 
\end{equation}
where $v_j\sim N(0,\tau _v^2)$ is a random effect. 
This is an example of the general model given in Section~\ref{sec:extension}. 
The model of this type is frequently used in longitudinal data analysis \citep[e.g.][]{verbeke2009linear}, where $m$ and $T$ are the numbers of subjects and repeated measurements, respectively, and $v_j$ is regarded as a subject-specific effect. 
Throughout this study, we set $m=50$, $T=10$ and $p=10$. 
We adopted the same values for $\beta_k$'s, and the same generating process for $(x_{jt1},\ldots,x_{jtp})$ and $\ep_{jt}$, as those in the previous simulation study. 
The scale parameters are set as $\tau^2_v=(0.5)^2$ and $\sigma=1$.

We model the distribution of error $\epsilon _{jt}$ in the model (\ref{RI}) by the EH distribution with latent variables $(z_{jt},u_{jt})$. The same data augmentation strategy can be used in the posterior computation for this model, and the full conditional distribution of $v_j$ is given by $N(\bt_j\at_j, \bt_j)$, where 
$$
\bt_j^{-1}=\frac{1}{\tau^2_v}+\frac1{\sigma^2}\sum_{t=1}^T\frac1{u_{jt}^{z_{jt}}}, \ \ \ \at_j=\frac{1}{\sigma^2}\sum_{t=1}^Tu_{jt}^{-z_{jt}}\Big(y_{tj}-\beta_0 -\sum_{k=1}^p \beta_kx_{jtk}\Big).
$$
We use an inverse-gamma prior for $\tau^2_v$, namely, $\tau^2_v\sim {\rm IG}(a_{v}, b_{v})$ with $a_v=b_v=1$, and the full conditional distribution of $\tau^2_v$ is ${\rm IG}(\at_{v},\bt_v)$, where $\at_{v}=a_{v}+m/2$ and $\bt_{v}=b_{v}+\sum_{j=1}^mv_j^2/2$.
Given the random effect $v_j$, the other parameters and latent variables can be easily generated from their full conditional distributions in Section~\ref{sec:post} with the slight modification by replacing the response variable with $y_{jt}-v_j$. 
The other error distributions, such as the normal and $t$-distributions and the finite mixture, can be implemented in the same way by using its representation of scale mixture of normals. 
The only exception is the LPTN distribution; it does not admit representation of scale mixture of normals and is not directly incorporated into the random intercept model.  
In total, we employed six error distributions (EH, aEH, C, aT, MT, N) in this study. 
We evaluated the performance of point and interval estimations by posterior means and 95\% credible intervals for the regression coefficients, using RMSE, CP and AL, as adopted in the previous study. 
The performance of the six models in predicting the random effect is also assessed via square root of mean squared prediction errors (RMSPE) based on 500 replications of the simulations, and these values are averaged over $v_1,\ldots,v_m$.

We report the results in Table \ref{tab:sim-RI}.
Regarding the regression coefficients, almost the same tendency as in Tables \ref{tab:sim} can be observed, which indicates the usefulness of the proposed EH method under more structured models than linear regression. 
It is also observed that the EH method with estimated $\gamma$ does not necessarily work well, thereby our recommendation in this example is simply using the fixed value $\gamma=1$.
In terms of RMSPE, the proposed EH method consistently outperforms the other methods. 
Specifically, the difference between the EH and MT methods is considerable, which also suggests the importance of the posterior robustness shown in Theorem~\ref{thm:normalized}, i.e., the advantage of the proposed error distribution over the conventional finite mixture approach with $t$-distribution.

\begin{table}[htbp]
\caption{Average values of RMSEs, CPs, ALs and RMSPEs of the proposed extremely-heavy tailed distribution with $\gamma$ fixed (EH) and estimated (aEH), Cauchy distribution (C), $t$-distribution with estimated degrees of freedom (aT), two component mixture of normal and $t$-distribution with $1/2$ degrees of freedom (MT), and normal distribution (N) under the random intercept models with 6 combinations of $(100\omega, \mu)$.
All values except for IFs are multiplied by 100.
\label{tab:sim-RI}}
\begin{center}
\begin{tabular}{ccccccccccccc}
\hline
 & $(100\omega, \mu)$ &  & EH & aEH & C & aT & MT & N \\
\hline
\multirow{6}{*}{RMSE} & (5, 5) &  & 5.91 & 5.89 & 6.86 & 6.54 & 5.91 & 10.67 \\
 & (10, 5) &  & 8.45 & 8.88 & 7.12 & 9.39 & 8.52 & 17.51 \\
 & (5, 10) &  & 5.61 & 5.58 & 6.84 & 6.37 & 5.72 & 19.40 \\
 & (10, 10) &  & 5.86 & 5.79 & 6.78 & 9.52 & 6.03 & 33.74 \\
 & (5, 15) &  & 5.47 & 5.45 & 6.65 & 6.10 & 5.58 & 28.23 \\
 & (10, 15) &  & 5.86 & 5.79 & 6.84 & 9.36 & 5.96 & 49.80 \\
\hline
\multirow{6}{*}{CP} & (5, 5) &  & 94.1 & 94.0 & 81.6 & 92.6 & 92.1 & 86.5 \\
 & (10, 5) &  & 92.9 & 93.2 & 84.5 & 90.4 & 91.1 & 85.7 \\
 & (5, 10) &  & 95.1 & 94.7 & 82.3 & 95.2 & 91.9 & 85.9 \\
 & (10, 10) &  & 95.2 & 95.2 & 86.0 & 95.3 & 91.8 & 86.0 \\
 & (5, 15) &  & 94.9 & 94.6 & 83.6 & 96.4 & 91.9 & 86.4 \\
 & (10, 15) &  & 95.5 & 95.4 & 84.8 & 97.3 & 92.5 & 86.5 \\
 \hline
\multirow{6}{*}{AL} & (5, 5) &  & 22.0 & 21.8 & 18.4 & 22.7 & 20.3 & 27.7 \\
 & (10, 5) &  & 25.8 & 26.0 & 19.7 & 28.5 & 23.1 & 33.7 \\
 & (5, 10) &  & 21.4 & 21.1 & 18.5 & 25.7 & 19.7 & 44.7 \\
 & (10, 10) &  & 23.1 & 22.7 & 19.6 & 38.9 & 21.0 & 58.5 \\
 & (5, 15) &  & 21.3 & 21.0 & 18.5 & 27.2 & 19.7 & 63.3 \\
 & (10, 15) &  & 23.0 & 22.6 & 19.6 & 47.2 & 20.9 & 85.1 \\
 \hline
\multirow{6}{*}{RMSPE} & (5, 5) &  & 29.5 & 29.4 & 33.9 & 31.5 & 35.0 & 40.5 \\
 & (10, 5) &  & 33.6 & 33.2 & 34.1 & 37.1 & 38.8 & 44.2 \\
 & (5, 10) &  & 28.8 & 28.8 & 34.2 & 33.3 & 33.7 & 46.9 \\
 & (10, 10) &  & 30.5 & 29.7 & 33.5 & 40.9 & 36.4 & 48.3 \\
 & (5, 15) &  & 28.8 & 28.8 & 34.0 & 33.9 & 33.6 & 48.3 \\
 & (10, 15) &  & 30.3 & 29.5 & 33.4 & 41.8 & 36.3 & 49.2 \\
\hline
\end{tabular}
\end{center}
\end{table}

\section{Real data examples}\label{sec:app}

The posterior robustness of the proposed EH distribution is demonstrated via the analysis of two real datasets: Boston housing data and diabetes data. The goal of statistical analysis here is the variable selection with $p=29$ and $p=64$ predictors in the presence/absence of outliers. Our robustness scheme is a prominent part of such analysis by allowing the use of unbounded prior densities for strong shrinkage effect-- specifically the horseshoe priors we discussed in Section~\ref{sec:shrink}-- while protecting the posteriors from the potential outliers. The former dataset is suspected to be contaminated with outliers, where the difference of the proposed EH distribution and the traditional $t$-distribution is emphasized. In contrast, the latter dataset is free from extreme outliers, and we use this dataset to discuss the possible efficiency loss caused by the use of EH distributions.

In our examples, we consider robust Bayesian inference using the proposed method with taking account of variable selection, since the number of covariates is not small in two cases. 
Specifically, we employed the horseshoe prior as described in Section~\ref{sec:shrink}.
For comparison, we also applied the standard normal distribution and the two-component mixture of normal and $t$-distributions as the error distribution, while using the horseshoe prior for regression coefficients. 
In all the methods, we generated 10000 posterior samples after discarding the first 5000 posterior samples as burn-in.

\subsection{Boston housing data}\label{sec:boston}

We first consider the famous Boston housing dataset \citep{Boston}. 
The response variable is the corrected median value of owner-occupied homes (in 1,000 USD). The covariates in the original datasets consist of 14 continuous-valued variables about the information of houses, such as per capita crime rate 
and accessibility to radial highways, and 1 binary covariate. 
After standardizing the 14 continuous covariates, we also create squared values of those, which results in $p=29$ covariates in our models. 
The sample size is $n=506$. 
The data also contains the longitude and latitude of house $i$, denoted by $t_i$. 
To take account of spatial correlation, we consider the following model: 
\begin{equation}\label{reg-sp}
y_i=x_i^t\beta + g(t_i) + \ep_i, \ \ \ i=1,\ldots,n,
\end{equation}
where $g(t_i)$ is a spatial effect as an unknown function of location information $t_i$.
We assume that $g(t_i)$ follows the standard Gaussian process, namely, $\eta\equiv (g(t_1),\ldots,g(t_n))$ and $\eta \sim N(0, \kappa^2 C(h))$, where $C(h)$ is a variance-covariance matrix whose $(i,j)$-entry is $\exp(-\|s_i-s_j\|^2/2h^2)$ with unknown bandwidth parameter $h$.
The above model can be seen as the spatially varying intercept model, or the spatially varying coefficient model \citep[e.g.][]{gelfand2003}. 
Also, this is another example of the general model in Section~\ref{sec:extension} with $r=n$, $b=\eta$, $g_i$ is the $i$-th standard basis, and $H(\psi) = \kappa ^2 C(h)$ with $\psi = (\kappa, h)$. 
Under the EH distribution for $\ep_i$, the full conditional distribution of $\eta$ is given by $N(\At_{\eta}^{-1}\Bt_{\eta}, \At_{\eta}^{-1})$, where 
$$
\At_{\eta}=\kappa^{-2}C(h)^{-1}+\sigma^{-2}{\rm diag}(u_1^{-z_1},\ldots,u_n^{-z_n}), \ \ \ {\rm and} \ \ \ 
\Bt_{\eta}=(Y-X\beta)/\sigma^2.
$$
A similar sampling strategy can be used for the two component mixture of a normal and $t$-distribution with $1/2$ degrees of freedom (denoted by MT), as adopted in the simulation study in Section \ref{sec:sim}. 
We employ the conjugate inverse gamma prior ${\rm IG}(1, 1)$ for $\tau^2$, and a uniform prior, $U(0,h_M)$, for $h$, where $h_M$ is the median of all the pair-wise distances of the sampling locations. 
The random-walk Metropolis-Hastings algorithm can be used for sampling from the full conditional distribution of $h$.

As the exploratory analysis, we first applied the model (\ref{reg-sp}) with normal error, $\ep_i\sim N(0,\sigma^2)$, and computed the standardized residuals by using the posterior mean of the model parameters to visualize the potential outliers. The computed residuals are shown in the left panel of Figure~\ref{fig:Boston}. 
Despite the normal error model is sensitive to outliers, there are still large residuals seen in the figure, which implies the extremity of the outliers in this dataset. 
In the proposed error distribution, the existence of extreme outliers is implied by the posterior of mixture weight $s$, i.e., the proportion of the extremely heavy-tailed distribution in the finite mixture. 
The trace plot of posterior samples of mixture weight $s$ under the EH model is presented in the right panel of Figure~\ref{fig:Boston}. 
Since all the sampled values are bounded away from $0$, it suggests that a certain proportion of the heavy-tailed distribution to take account of the outliers shown in the left panel. 
As the prior sensitivity analysis, we also applied more informative priors, ${\rm Beta}(1,5)$ and ${\rm Beta}(1,9)$, in addition to the default prior $s\sim {\rm Beta}(1,1)$, based on the prior belief that $s$ should be small. 
However, the posteriors computed with the three beta priors are almost identical.

The estimated spatial effects, $g(t_i)$, under the EH and normal models are presented in Figure~\ref{fig:Boston-sp}.
The EH model produces spatially smoothed estimates, while the estimates of the normal model are volatile across the sampling area. This finding also evidences the effect of outliers on the posterior inference for the regression coefficients or, in this example, the random intercept terms. 

The posterior means and $95\%$ credible intervals of the regression coefficients based on the three methods are shown in Figure~\ref{fig:Boston-coef}. 
It shows that the results of the normal error model are quite different from those of the MT and EH distributions. The difference of estimates becomes visually clear especially for the significant covariates-- if we define the significance in the sense that the $95\%$ credible intervals do not contain zero-- as the result of proneness/sensitivity to the representative outliers observed in Figure~\ref{fig:Boston}. 
The difference between the posteriors of the EH and MT models does exist, but is not as visually clear as the difference from the normal error model.

Finally, we computed the deviance information criterion \citep{spiegelhalter2002bayesian} of the three models.
The obtained values were 2628 for the normal error model, 2339 for the MT error model, and 2325 for the proposed EH error model, which shows the best fit of the EH error model to the data.

\begin{figure}[!htb]
\centering
\includegraphics[width=14cm,clip]{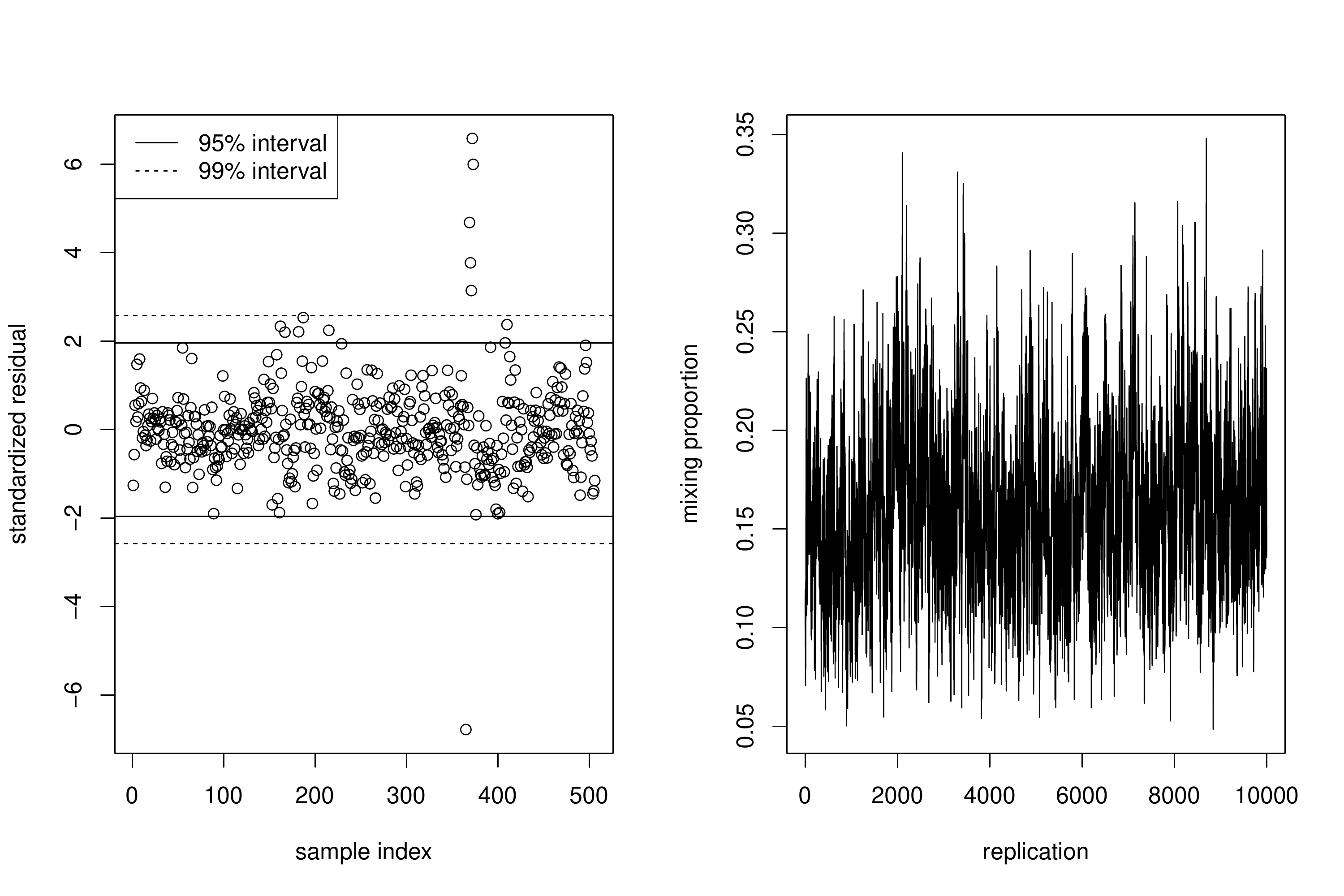}
\caption{Standardized residuals (left) and trace plot of $s$ (mixing proportion) in the proposed EH distribution (right), obtained form the Boston housing data.
The posterior mean and the $95\%$ credible interval of $s$ are $0.160$ and $(0.087, 0.249)$, respectively.  
\label{fig:Boston}
}
\end{figure}

\begin{figure}[!htb]
\centering
\includegraphics[width=13cm,clip]{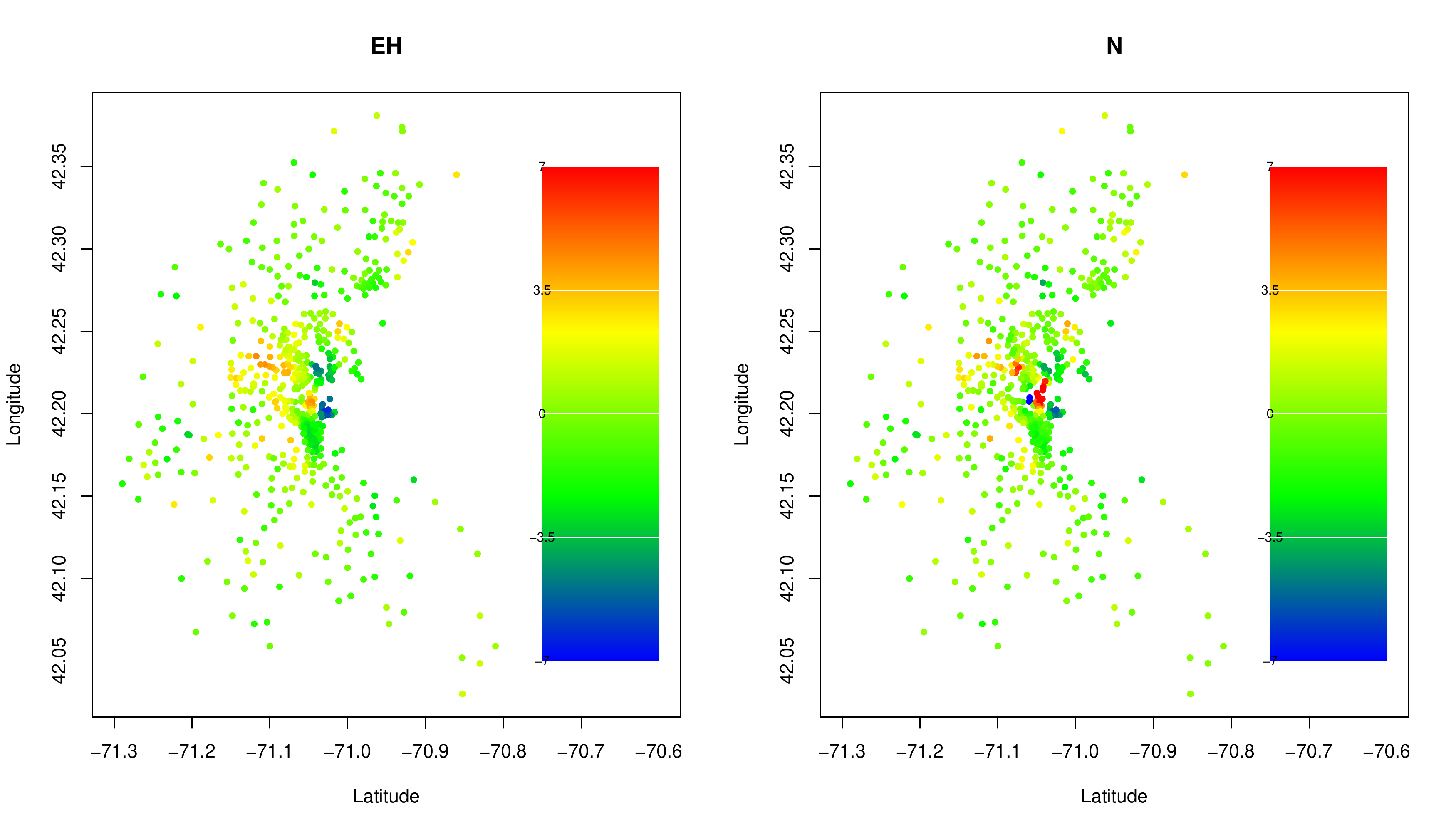}
\caption{Posterior means of the spatial effects based on the EH and the normal (N) distribution. 
\label{fig:Boston-sp}
}
\end{figure}

\begin{figure}[!htb]
\centering
\includegraphics[width=12cm,clip]{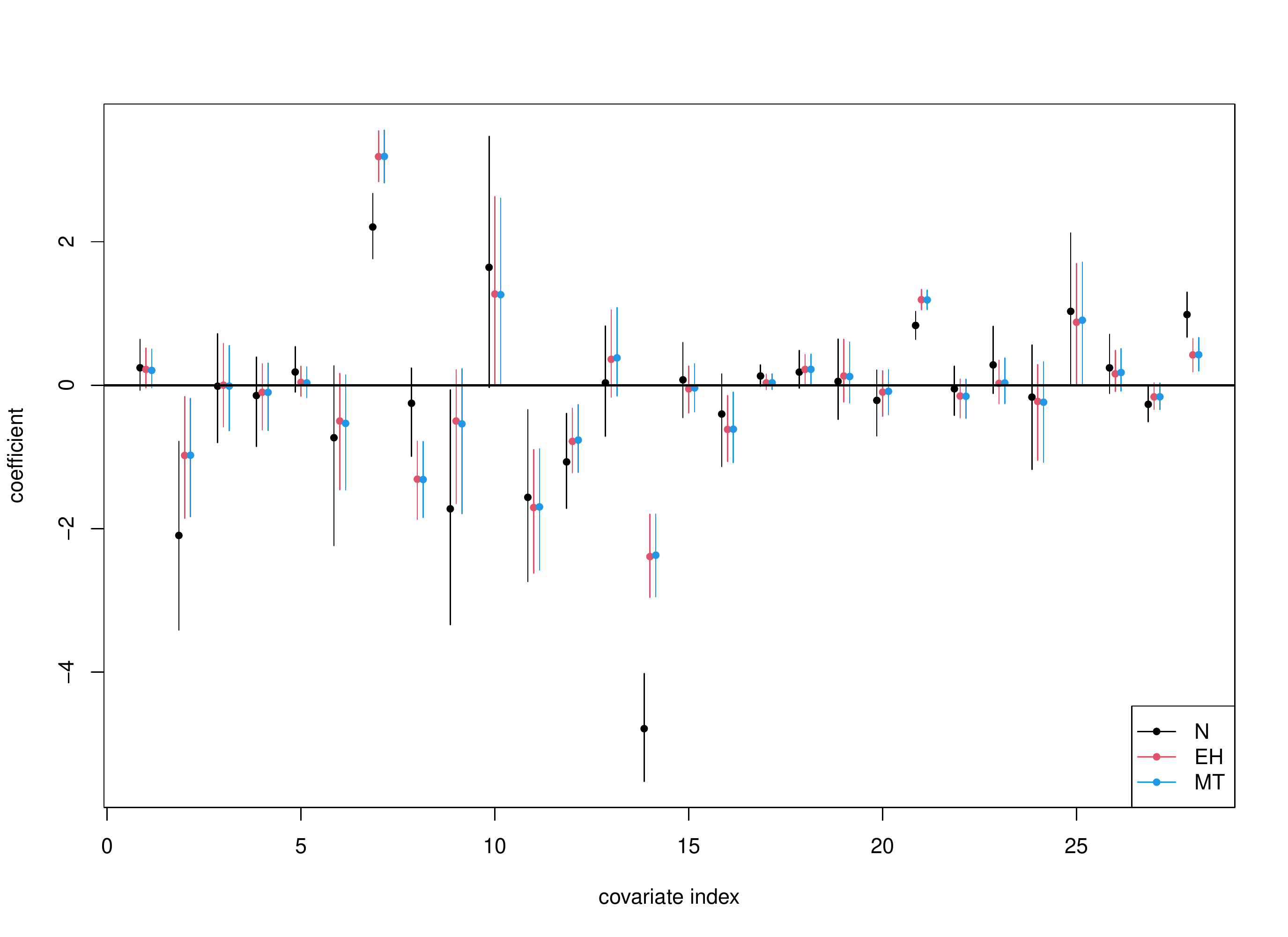}
\caption{Posterior means and $95\%$ credible intervals of the regression coefficients in the normal regression with normal distribution error (N), the proposed EH distribution, and the tow-component mixture of normal and $t$-distribution with $1/2$ degrees of freedom (MT), applied to the Boston housing data. 
\label{fig:Boston-coef}
}
\end{figure}

\subsection{Diabetes data}
We next consider another famous dataset known as Diabetes data \citep{Efron2004}.
The data contains information of $442$ individuals and 10 covariates regarding the personal information and related medical measures of the individuals.
We consider the same formulation of linear regression model as in \cite{Efron2004}; the set of predictors consists of the original 10 variables, 45 interactions, and 9 squared values, which results in $p=64$ predictors in the model.
For this dataset, the regression models with horseshoe prior and three error distributions (N, EH and MT) adopted in Section \ref{sec:boston} are applied.

Similarly to the analysis of Boston housing data, we check the standardized residuals computed under the standard linear regression model, which was presented in the left panel of Figure~\ref{fig:Diabetes}.
Few outliers are confirmed in the dataset as most of residuals are contained in the $99\%$ interval, which strongly supports the standard normal assumption in this example. 
The right panel of Figure~\ref{fig:Diabetes} shows the trace plot of posterior samples of mixture $s$ under the EH distribution. 
All the sampled values are very close to zero, implying that most error terms should be generated from the first component of the mixture, i.e., the standard normal distribution. In this case, the heavy-tailed component might be regarded ``redundant'' for this dataset. 
The same sensitivity analysis on the choice of priors for $s$ is done as in the previous section, but we find no significant change to the results.

To see the possible inefficiency of using the EH models for the dataset without outliers, the posterior means and $95\%$ credible intervals of the regression coefficients are reported in Figure~\ref{fig:Diabetes-coef}. 
The results of the three models are comparable; the predictors selected by significance are almost the same under the three models. 
The only notable difference is that  the credible intervals produced by the $t$-distribution model is slightly larger than those of the other two methods. This indicates the loss of efficiency in using the $t$-distribution method under no outliers, as also confirmed in the simulation results in Section~\ref{sec:sim}. In contrast, the difference in the credible intervals of the Gaussian and EH models is hardly visible in the figure. 
That is, even if no outlier exists, the efficiency loss in estimation under the EH model is minimal.

We also computed the deviance information criterion of the three models.
The obtained values were 4794 for the normal error model and 4795 for both the MT and EH error models, which shows the comparable fit of the three models.

\begin{figure}[!htb]
\centering
\includegraphics[width=14cm,clip]{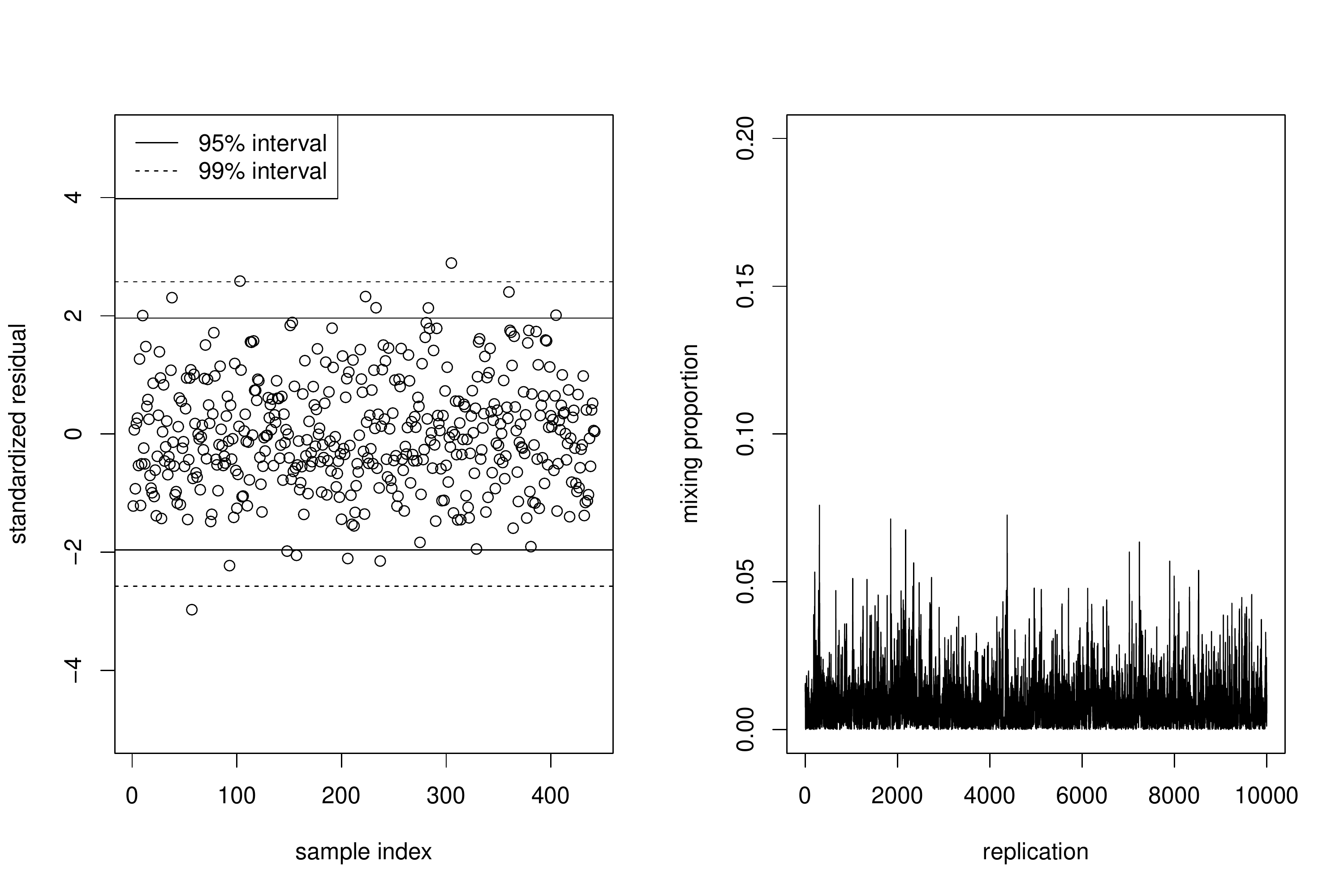}
\caption{Standardized residuals (left) and trace plot of $s$ (mixing proportion) in the proposed EH distribution (right), obtained form the Diabetes data.
The posterior mean and the $95\%$ credible interval of $s$ are $0.008$ and $(0.000, 0.032)$, respectively.  
\label{fig:Diabetes}
}
\end{figure}

\begin{figure}[!htb]
\centering
\includegraphics[width=12cm,clip]{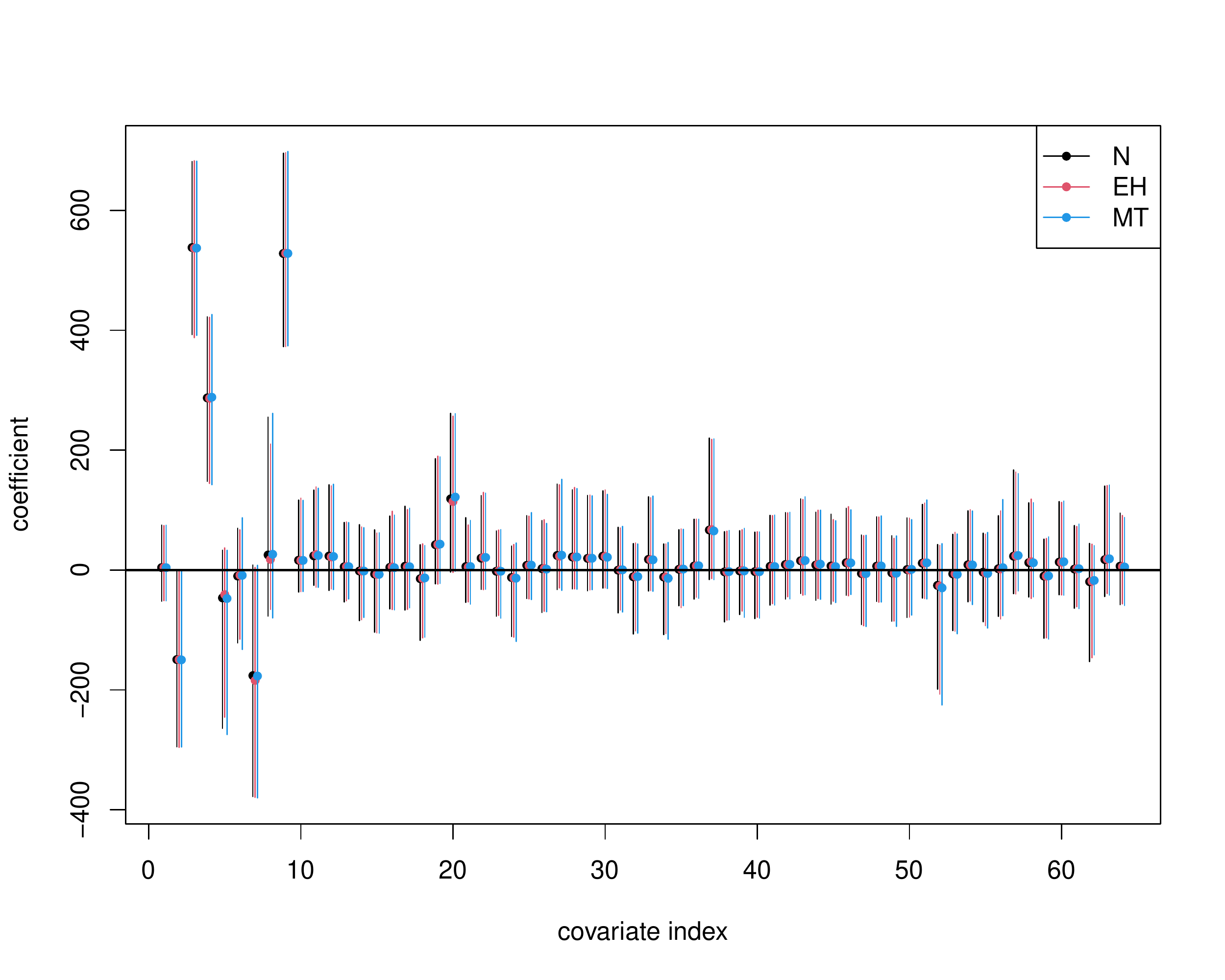}
\caption{Posterior means and $95\%$ credible intervals of the regression coefficients in the normal regression with normal distribution error (N), the proposed EH distribution, and the $t$-distribution (T) with estimated degrees of freedom, applied to the Diabetes data. 
\label{fig:Diabetes-coef}
}
\end{figure}

\section*{Acknowledgement}
This work is supported by the Japan Society for the Promotion of Science (grant number: 18K12757, 17K17659 and 18H00835).

\section{Discussions}\label{sec:disc}

While the focus of this research is on the inference for the regression coefficients and scale parameter, it is also of great interest to employ the predictive analysis based on the proposed model. Because $H$-distribution, as well as many log-regularly varying distributions, is too heavily-tailed to have finite moments, the posterior predictive moments under the EH models do not exist. It is common in practice to have predictive distributions with no finite moments \citep{West2020decisionconstraints} and it is worth investigating the predictive properties under the EH models, especially about the impact of the heavy tails on the predictive uncertainty.

The proposed method is not limited to the analysis of the linear regression models, but can be immediately customized for any conditionally Gaussian models, as we practiced in the analysis of the random intercept model in Section~\ref{sec:sim-RI} and the spatially varying intercept model in Section~\ref{sec:boston}. Other examples include graphical models and dynamic linear models, which can be the topics of the promising future research. 
The efficient posterior computation algorithm presented in this research can be used for these highly-structured models as well by utilizing the hierarchical representation of the proposed error distribution. 
The similar theoretical robustness properties may also be confirmed for those models.

Finally, we note that the assumption (A.1) in Theorem \ref{thm:normalized} misses the high-dimensional regression with small sample size ($n<p$), which means that posterior robustness is not necessarily achieved in this challenging situation. 
Therefore, substantial work will be required to develop the theory and methodology for ``robust high-dimensional regression,'' which we left to an interesting future research topic.

\section*{Supplementary Material}
Proofs of all the propositions and theorems, and additional simulation results are given in the online supplementary material.


\newpage
\begin{center}
{\Large
{\bf 
Supplemental Materials for\\
 ``Log-Regularly Varying Scale Mixture of Normals for Robust Regression''
}
}
\end{center}

\setcounter{equation}{0}
\renewcommand{\theequation}{S\arabic{equation}}
\setcounter{section}{0}
\renewcommand{\thesection}{S\arabic{section}}
\setcounter{table}{0}
\renewcommand{\thetable}{S\arabic{table}}
\setcounter{figure}{0}
\renewcommand{\thefigure}{S\arabic{figure}}
\setcounter{lem}{0}
\renewcommand{\thelem}{S\arabic{lem}}
\setcounter{thm}{0}
\renewcommand{\thethm}{S\arabic{thm}}
\setcounter{prp}{0}
\renewcommand{\theprp}{S\arabic{prp}}

\vspace{1cm}
This Supplementary Material provides proofs of Proposition~2.1, Theorem~2.1, Proposition~2.2 and Corollary~2.1 and additional simulation results.

\section{Lemmas}
In this section, we provide lemmas used in the proofs.

\begin{lem}
\label{lem:IL} 
Let $M, v > 0$. Then we have
\begin{align}
(a) \hspace{30pt} &{1 + \log (1 + M) \over 1 + \log (1 + M v)} \le \max \{ 1, v^{- 1} \} \text{,} \non \\
& \non \\
(b) \hspace{30pt} &\lim_{M \to \infty } {1 + \log (1 + M) \over 1 + \log (1 + M v)} = 1 \text{.} \non 
\end{align}
\end{lem}

\begin{proof}
The inequality in part (a) is trivial when $v \ge 1$; the left-hand-side is bounded by $1$. For the case of $v < 1$, first observe that 
\begin{align}
{1 + \log (1 + M) \over 1 + \log (1 + M v)} &= \exp \Big( \int_{v}^{1} \Big[ {\partial \over \partial t} \log \{ 1 + \log (1 + M t) \} \Big] dt \Big) \non \\
&= \exp \Big\{ \int_{v}^{1} {1 \over 1 + \log (1 + M t)} {M \over 1 + M t} dt \Big\} \non 
\end{align}
for all $v > 0$. 
Then it is immediate from this expression that 
\begin{align}
{1 + \log (1 + M) \over 1 + \log (1 + M v)} &\le \exp \Big( \int_{v}^{1} {1 \over t} dt \Big) = v^{- 1} \non 
\end{align}
for $v < 1$. For part (b), we use the same expression to obtain 
\begin{align}
&\lim _{M\to \infty} {1 + \log (1 + M) \over 1 + \log (1 + M / v)} \non \\ 
&= \exp \Big\{ \lim _{M\to \infty} \int_{v}^{1} {1 \over 1 + \log (1 + M t)} {M \over 1 + M t} dt \Big\} \non \\
&= 1 \non 
\end{align}
by the dominated convergence theorem.
\end{proof}

\begin{lem}
\label{lem:asymptotic} 
For $\ga > 0$ and $\de \ge 0$, let $H(u; \ga , \de )$, $u \in (0, \infty )$, be the proper density proportional to 
\begin{align}
H(u; \ga , \de ) &\propto {1\over (1 + u)^{1 + \delta}} {1 \over \{ 1 + \log (1 + u) \} ^{1 + \ga }} \text{.} \non 
\end{align}
Then we have 
\begin{align}
\int_{0}^{\infty } {\rm{N}} (z | 0, u) H(u; \ga , \de ) du \sim {C^{- 1} \over \sqrt{2 \pi }} \Big( {z^2 \over 2} \Big) ^{- 1 / 2 - \de } \Ga \Big( {1 \over 2} + \de \Big) S \Big( {z^2 \over 2} \Big) \non 
\end{align}
as $|z| \to \infty $, where 
\begin{align}
S(u) &= \Big( {u \over 1 + u} \Big) ^{1 + \de } {1 \over \{ 1 + \log (1 + u) \} ^{1 + \ga }} \non 
\end{align}
for $u \in (0, \infty )$ and $C = \int_{0}^{\infty } u^{- \de - 1} S(u) du$. 
\end{lem}

\begin{proof}
The result follows by (the proof of) part 3 of Theorem 1 of Hamura et al. (2020). 
\end{proof}

\begin{lem}
\label{lem:f0f1f} 
For $z \in \mathbb{R}$, let $f_0 (z) = {\rm{N}} (z | 0, 1)$ and 
\begin{align}
f_1 (z) &= \int_{0}^{\infty } {\rm{N}} (z | 0, u) H(u; \ga ) du \text{,} \non 
\end{align}
where $H(u; \ga ) = H(u; \ga , \de = 0)$, $u \in (0, \infty )$, is the H-distribution.  
\begin{enumerate}
\item[{\rm{(i)}}]
For all $z \in \mathbb{R}$, we have $0 < f_0 (z) = f_0 (|z|) \le f_0 (0) < \infty $, $0 < f_1 (z) = f_1 (|z|) \le f_1 (0) < \infty $, and $0 < f(z) = f(|z|) \le f(0) < \infty $. 
\item[{\rm{(ii)}}]
$f_0 (r)$, $f_1 (r)$, and $f(r)$ are continuous and nonincreasing functions of $r$ for $r \ge 0$. 
\item[{\rm{(iii)}}]
$\lim_{|z| \to \infty } \{ f_0 (z) / f_1 (z) \} = 0$ and $\sup_{z \in \mathbb{R}} \{ f_0 (z) / f_1 (z) \} < \infty $. 
\item[{\rm{(iv)}}]
There exists $C_1 > 0$ such that for all $z \neq 0$, we have 
\begin{align}
f_1 (z) \le {C_1 \over |z|} {1 \over \{ 1 + \log (1 + |z|^2 ) \} ^{1 + \ga }} \text{.} \non 
\end{align}
\item[{\rm{(v)}}]
There exists $C_2 > 0$ such that for all $z \in \mathbb{R} \setminus (- 1, 1)$, we have 
\begin{align}
f_1 (z) \ge {C_2 \over |z|} {1 \over \{ 1 + \log (1 + |z|^2 ) \} ^{1 + \ga }} \text{.} \non 
\end{align}
\item[{\rm{(vi)}}]
There exists $C_3 > 0$ such that for all $z \in \mathbb{R}$, we have 
\begin{align}
f_1 (z) \le {C_3 \over 1 + |z|} \text{.} \non 
\end{align}
\item[{\rm{(vii)}}]
There exists $C_4 > 0$ such that for all $\yt , \mut \in \mathbb{R}$ and all $\si \in (0, \infty )$, we have 
\begin{align}
| \yt | f_1 \Big( {\yt - \mut \over \si } \Big) &\le | \yt | f_1 \Big( {| \yt | - | \mut | \over \si } \Big) \le C_4 ( \si + | \mut |) \text{.} \non 
\end{align}
\end{enumerate}
\end{lem}

\begin{proof}
Parts (i) and (ii) are trivial. 
Since $\lim_{z \to 0} \{ f_0 (z) / f_1 (z) \} = f_0 (0) / f_1 (0) < \infty $ by parts (i) and (ii) and since $\lim_{|z| \to \infty } \{ f_0 (z) / f_1 (z) \} = 0$ by Lemma \ref{lem:asymptotic}, part (iii) follows. 

Note that for all $z \neq 0$, 
\begin{align}
f_1 (z) &= \int_{0}^{\infty } {\rm{N}} (z | 0, u) H(u; \ga ) du \non \\
&= \int_{0}^{\infty } {1 \over \sqrt{2 \pi }} {1 \over u^{1 / 2}} e^{- z^2 / (2 u)} {\ga \over 1 + u} {1 \over \{ 1 + \log (1 + u) \} ^{1 + \ga }} du \non \\
&= \int_{0}^{\infty } {|z| \over \sqrt{2 \pi }} {1 \over u^{1 / 2}} e^{- 1 / (2 u)} {\ga \over 1 + |z|^2 u} {1 \over \{ 1 + \log (1 + |z|^2 u) \} ^{1 + \ga }} du \non \\
&= {\ga |z| \over \sqrt{2 \pi }} {1 \over \{ 1 + \log (1 + |z|^2 ) \} ^{1 + \ga }} \int_{0}^{\infty } {1 \over u^{1 / 2}} e^{- 1 / (2 u)} {1 \over 1 + |z|^2 u} \Big\{ {1 + \log (1 + |z|^2 ) \over 1 + \log (1 + |z|^2 u)} \Big\} ^{1 + \ga } du \text{.} \non 
\end{align}
Then, by part (a) of Lemma \ref{lem:IL}, we have 
\begin{align}
f_1 (z) &\le {\ga |z| \over \sqrt{2 \pi }} {1 \over \{ 1 + \log (1 + |z|^2 ) \} ^{1 + \ga }} {1 \over |z|^2} \int_{0}^{\infty } {1 \over u^{3 / 2}} e^{- 1 / (2 u)} ( \max \{ 1, 1 / u \} )^{1 + \ga } du \non \\
&\le {\ga \over \sqrt{2 \pi }} {1 \over |z|} {1 \over \{ 1 + \log (1 + |z|^2 ) \} ^{1 + \ga }} \int_{0}^{\infty } \Big\{ {e^{- 1 / (2 u)} \over u^{3 / 2}} + {e^{- 1 / (2 u)} \over u^{5 / 2 + \ga }} \Big\} du < \infty \non 
\end{align}
for all $z \neq 0$, which prove part (iv). 
Also, for all $z \in \mathbb{R} \setminus (- 1, 1)$, 
\begin{align}
f_1 (z) &\ge {\ga |z| \over \sqrt{2 \pi }} {1 \over \{ 1 + \log (1 + |z|^2 ) \} ^{1 + \ga }} {1 \over |z|^2} \int_{0}^{\infty } {1 \over u^{1 / 2}} e^{- 1 / (2 u)} {1 \over 1 + u} ( \min \{ 1, 1 / u \} )^{1 + \ga } du \non 
\end{align}
by part (a) of Lemma \ref{lem:IL}. 
This proves part (v). 
It follows from parts (i) and (iv) that $f_1 (z) \le \max \{ 2 f_1 (0) / (1 + |z|), 2 C_1 / (1 + |z|) \} $ for all $z \in \mathbb{R}$. 
Therefore, part (vi) is obtained. 
For part (vii), let $\yt , \mut \in \mathbb{R}$ and $\si \in (0, \infty )$. 
Then $f_1 (( \yt - \mut ) / \si ) \le f_1 ((| \yt | - | \mut |) / \si )$ by parts (i) and (ii). 
Suppose first that $| \yt | \le \si + | \mut |$. 
Then $| \yt | f_1 ((| \yt | - | \mut |) / \si ) \le f_1 (0) ( \si + | \mut |)$ by part (i). 
Next, suppose that $| \yt | > \si + | \mut |$. 
Then, since $| \yt | - | \mut | \ge 0$ and $(| \yt | - | \mut |) / \si \ge 1$, 
\begin{align}
&| \yt | f_1 \Big( {| \yt | - | \mut | \over \si } \Big) \non \\
&= {\ga \over \sqrt{2 \pi }} \int_{0}^{\infty } {1 \over u^{1 / 2}} e^{- 1 / (2 u)} {| \yt | (| \yt | - | \mut |) / \si \over 1 + \{ (| \yt | - | \mut |) / \si \} ^2 u} {1 \over (1 + \log [1 + \{ (| \yt | - | \mut |) / \si \} ^2 u])^{1 + \ga }} du \non \\
&\le {\ga \over \sqrt{2 \pi }} \int_{0}^{\infty } {1 \over u^{1 / 2}} e^{- 1 / (2 u)} {1 \over \si } {(| \yt | - | \mut |)^2 + | \mut | (| \yt | - | \mut |) \over 1 + \{ (| \yt | - | \mut |) / \si \} ^2 u} du \non \\
&\le {\ga \over \sqrt{2 \pi }} \int_{0}^{\infty } {1 \over u^{1 / 2}} e^{- 1 / (2 u)} {1 \over \si } \Big( {\si ^2 \over u} + {\si | \mut | \over u} \Big) du = {\ga \over \sqrt{2 \pi }} ( \si + | \mut | ) \int_{0}^{\infty } {1 \over u^{3 / 2}} e^{- 1 / (2 u)} du < \infty \text{.} \non 
\end{align}
This completes the proof of part (vii). 
\end{proof}

\begin{lem}
\label{lem:linear} 
Let $m, p \in \mathbb{N}$. 
Let $w = ( w_1 , \dots , w_m )^t \in \mathbb{R} ^m$. 
Let $Z = ( z_1 , \dots , z_m )^t$ be an $m \times p$ matrix of observations (such that any set of its $p$ distinct row vectors is linearly independent). 
Suppose that $m \ge p$. 
Then there exist $R > 0$ and $\de > 0$ (which may depend on $w$ and $Z$) such that 
\begin{align}
\prod_{i = 1}^{m} {1 \over 1 + | w_i - z_{i}^{t} \be |} \le {1 \over (1 + \de | \be | )^{m - p + 1}} \non 
\end{align}
for all $\be \in \mathbb{R} ^p$ satisfying $| \be | \ge R$. 
\end{lem}

\begin{proof}
In this proof, if $A$ is a matrix, we write $|A| = \sqrt{\tr ( A^t A)}$. 
Let $I = \{ ( i_k )_{k = 1}^{p} | 1 \le i_1 < \dots < i_p \le m \} $. 
For $i = ( i_k )_{k = 1}^{p} \in I$, let $w(i) = ( w_{i_1} , \dots , w_{i_p} )^t$ and $Z(i) = ( z_{i_1} , \dots , z_{i_p} )^t$. 
Let $R = 1 + 2 \max_{i \in I} | {Z(i)}^{- 1} w(i)| > 0$ and $\de = \min_{i \in I} \{ 1 / (2 \sqrt{p} | {Z(i)}^{- 1} | ) \} > 0$. 
Let $\be \in \mathbb{R} ^p$ be such that $| \be | \ge R$. 
Then for all $i = ( i_k )_{k = 1}^{p} \in I$, we have that $| \be - {Z(i)}^{- 1} w(i) | \ge | \be | - R / 2 \ge | \be | / 2$ and hence that $| \be | \le 2 | \be - {Z(i)}^{- 1} w(i) | \le 2 | {Z(i)}^{- 1} | | Z(i) \be - w(i) | \le 2 \sqrt{p} | {Z(i)}^{- 1} | \max_{1 \le k \le p} | z_{i_k}^{t} \be - w_{i_k} |$, which implies that there exists $k = 1, \dots , p$ such that $\de | \be | \le | z_{i_k}^{t} \be - w_{i_k} |$. 
Therefore, we can choose distinct indices $i^{(1)} , \dots , i^{(m - p + 1)} = 1, \dots , m$ so that for all $j = 1, \dots , m - p + 1$, we have $\de | \be | \le | z_{i^{(j)}}^{t} \be - w_{i^{(j)}} |$. 
Indeed, for $j = 2, \dots , m - p + 1$, given $i^{(1)} , \dots , i^{(j - 1)}$, we can choose $i_1 < \dots < i_p$ from $\{ 1, \dots , m \} \setminus \{ i^{(1)} , \dots , i^{(j - 1)} \} $ and then $k$ with $\de | \be | \le | z_{i_k}^{t} \be - w_{i_k} |$ from $\{ 1, \dots , p \} $ and set $i^{(j)} = i_k$. 
Thus, $\prod_{i = 1}^{m} (1 + | w_i - z_{i}^{t} \be |) \ge \prod_{j = 1}^{m - p + 1} (1 + | w_{i^{(j)}} - z_{i^{(j)}}^{t} \be |) \ge (1 + \de | \be | )^{m - p + 1}$. 
\end{proof}

\begin{lem}
\label{lem:ratio} 
Let $\al ( \cdot )$ and $\be ( \cdot )$ be continuous, positive, and integrable functions defined on $(0, \infty )$. 
Suppose that $\lim_{u \to \infty } \be (u) / \al (u) = \rho \in [0, \infty ]$. 
Then 
\begin{align}
\lim_{z \to \infty } \int_{0}^{\infty } {\rm{N}} (z | 0, u) \be (u) du \Big/ \int_{0}^{\infty } {\rm{N}} (z | 0, u) \al (u) du = \rho \text{.} \non 
\end{align}
\end{lem}

\begin{proof}
We can assume that $\rho < \infty $; if $\rho = \infty$, then we can exchange the definitions of $\alpha (\cdot )$ and $\beta (\cdot )$, and this reduces to the case of $\rho = 0$. 
Let $\ga (\cdot )$ be either $\al (\cdot )$ or $\be (\cdot )$. 
We can also assume without loss of generality that $u^{- 1 / 2} \al (u)$ and $u^{- 1 / 2} \be (u)$ are integrable. To see this, observe that, for any $\eta > 0$, there exist $\ep > 0$ satisfying 
\begin{align}
0 \le \frac{ \int_{0}^{\ep } {\rm{N}} (1 | 0, u) \ga (u) du }{ \int_{0}^{\infty } {\rm{N}} (1 | 0, u) \ga (u) du } < \eta / 2 \non 
\end{align}
and, for these $\eta$ and $\epsilon$, there also exists $\delta >0$ such that $0 \le 1 - e^{- \de / \ep } < \eta / 2$. Hence, for all $z \ge 1$, the covariance inequality implies
\begin{align}
\frac{ \int_{0}^{\ep } {\rm{N}} (z | 0, u) \ga (u) du }{ \int_{0}^{\infty } {\rm{N}} (z | 0, u) \ga (u) du } &= E[ \chi _{(0, \ep )} ( U_z ) ] \non \\
&\le {E[ \exp \{ ( z^2 - 1) / (2 U_z ) \} \chi _{(0, \ep )} ( U_z ) ] \over E[ \exp \{ ( z^2 - 1) / (2 U_z ) \} ]} \non \\
&= \frac{ \int_{0}^{\ep } {\rm{N}} (1 | 0, u) \ga (u) du }{ \int_{0}^{\infty } {\rm{N}} (1 | 0, u) \ga (u) du } \non 
\end{align}
where $\chi _{(0, \ep )} (x)$ is the indicator function ($\chi _{(0, \ep )} (x) = 1$ if $x\in (0, \ep )$ and $0$ otherwise) and the density of random variable $U_z$ is proportional to ${\rm{N}} (z | 0, u) \ga (u)$. Finally, we have 
\begin{align}
\Big| \frac{ \int_{0}^{\infty } {\rm{N}} (z | 0, u) \ga (u) e^{- \de / u} du }{ \int_{0}^{\infty } {\rm{N}} (z | 0, u) \ga (u) du } - 1 \Big| &\le \frac{ \int_{0}^{\ep } {\rm{N}} (z | 0, u) \ga (u) du }{ \int_{0}^{\infty } {\rm{N}} (z | 0, u) \ga (u) du } + \frac{ \int_{\ep }^{\infty } {\rm{N}} (z | 0, u) \ga (u) (1 - e^{- \de / u} ) du }{ \int_{\ep }^{\infty } {\rm{N}} (z | 0, u) \ga (u) du } \non \\
&\le \frac{ \int_{0}^{\ep } {\rm{N}} (1 | 0, u) \ga (u) du }{ \int_{0}^{\infty } {\rm{N}} (1 | 0, u) \ga (u) du } + 1 - e^{- \de / \ep } \non \\
&< \eta \text{,} \non 
\end{align}
which shows the difference of $\ga (u)$ and $e^{-\delta /u} \ga (u)$ is ignorable in $u\to \infty$. This result verifies that, if $u^{-1/2}\ga (u)$ is not integrable, then we can replace $\ga (u)$ by $e^{-\delta /u} \ga (u)$. 

Again, assume $\rho < \infty$ and both $u^{-1/2}\alpha (u)$ and $u^{-1/2}\beta (u)$ are integrable. Let $M > 0$. Then we have 
\begin{align}
\Big| \frac{ \int_{0}^{\infty } {\rm{N}} (z | 0, u) \ga (u) du }{ \int_{M}^{\infty } {\rm{N}} (z | 0, u) \ga (u) du } - 1 \Big| &\le \frac{ \int_{0}^{M} {\rm{N}} (z | 0, u) \ga (u) du }{ \int_{M + 1}^{\infty } {\rm{N}} (z | 0, u) \ga (u) du } \non \\
&\le \Big\{ {e^{1 / (M + 1)} \over e^{1 / M}} \Big\} ^{z^2 / 2} \frac{ \int_{0}^{M} u^{- 1 / 2} \ga (u) du }{ \int_{M + 1}^{\infty } u^{- 1 / 2} \ga (u) du } \non \\
&\to 0 \non 
\end{align}
as $z \to \infty $ since $u^{- 1 / 2} \ga (u)$ is assumed to be integrable on $(0, \infty )$. 
Therefore, 
\begin{align}
\frac{ \int_{0}^{\infty } {\rm{N}} (z | 0, u) \be (u) du }{ \int_{0}^{\infty } {\rm{N}} (z | 0, u) \al (u) du } \approx \frac{ \int_{M}^{\infty } {\rm{N}} (z | 0, u) \be (u) du }{ \int_{M}^{\infty } {\rm{N}} (z | 0, u) \al (u) du } \label{lratiop1} 
\end{align}
as $z \to \infty $. 
Furthermore, uniformly in $z$, 
\begin{align}
\Big| \frac{ \int_{M}^{\infty } {\rm{N}} (z | 0, u) \be (u) du }{ \int_{M}^{\infty } {\rm{N}} (z | 0, u) \al (u) du } - \rho \Big| &\le \frac{ \int_{M}^{\infty } | \be (u) / \al (u) - \rho | {\rm{N}} (z | 0, u) \al (u) du }{ \int_{M}^{\infty } {\rm{N}} (z | 0, u) \al (u) du } \non \\
&\le \sup_{u > M} \Big| {\be (u) \over \al (u)} - \rho \Big| \non \\
&\to 0 \label{lratiop2} 
\end{align}
as $M \to \infty $ by assumption. 
Combining (\ref{lratiop1}) and (\ref{lratiop2}) gives the desired result. 
\end{proof}

\section{Proof of Proposition 2.1}

Here we prove Proposition 2.1. 
We show that 
\begin{align}
\lim_{|x| \to \infty } {f_{\rm EH} (x) \over |x|^{-1}(\log|x|)^{-1-\gamma}} = A \non 
\end{align}
for some constant $A > 0$. 
Since 
\begin{align}
\lim_{|x| \to \infty } \frac{ {\rm{N}} (x | 0, 1) }{ \int_{0}^{\infty} {\rm{N}} (x | 0, u) H(u; \ga ) du } &= 0 \non 
\end{align}
by part (iii) of Lemma \ref{lem:f0f1f}, we can assume $s = 1$. 
Then we have for sufficiently large $|x|$ 
\begin{align}
{f_{\rm EH} (x) \over |x|^{-1} (\log|x|)^{-1-\gamma}} &= \int_{0}^{\infty} {{\rm{N}} (x | 0, u) H(u; \ga ) \over |x|^{-1}(\log|x|)^{-1-\gamma}} du \non \\
&= \int_{0}^{\infty} {1 \over \sqrt{2 \pi }} {1 \over \sqrt{u}} e^{- x^2 / (2 u)} {\ga |x| \over 1 + u} \Big\{ {\log |x| \over 1 + \log (1 + u)} \Big\} ^{1 + \ga } du \non \\
&= \int_{0}^{\infty} {1 \over \sqrt{2 \pi }} {1 \over \sqrt{v}} e^{- 1 / (2 v)} {\ga x^2 \over 1 + x^2 v} \Big\{ {\log |x| \over 1 + \log (1 + x^2 v)} \Big\} ^{1 + \ga } dv \text{,} \non 
\end{align}
where the last equality follows by making the change of variables $u = x^2 v$. 
Now, by part (a) of Lemma \ref{lem:IL}, the integrand is bounded by 
\begin{align}
&{1 \over \sqrt{2 \pi }} {1 \over \sqrt{v}} e^{- 1 / (2 v)} {\ga \over v} \Big\{ {\log |x| \over 1 + \log (1 + x^2 )} {1 + \log (1 + x^2 ) \over 1 + \log (1 + x^2 v) } \Big\} ^{1 + \ga } \non \\
&\le {\ga \over \sqrt{2 \pi }} {e^{- 1 / (2 v)} \over v^{3 / 2}} \Big( {1 \over 2} \max \{ 1, v^{- 1} \} \Big) ^{1 + \ga } = {\ga / 2^{1 + \ga } \over \sqrt{2 \pi }} {e^{- 1 / (2 v)} \over v^{3 / 2}} \max \{ 1, v^{- (1 + \ga )} \} \non \\
&\le {\ga / 2^{1 + \ga } \over \sqrt{2 \pi }} \{ v^{- 3 / 2} e^{- 1 / (2 v)} + v^{- 5 / 2 - \ga } e^{- 1 / (2 v)} \} \text{,} \non 
\end{align}
where the right-hand side is an integrable function of $v \in (0, \infty )$ which does not depend on $x$. 
By part (b) of Lemma \ref{lem:IL}, the integrand converges to 
\begin{align}
{1 \over \sqrt{2 \pi }} {1 \over \sqrt{v}} e^{- 1 / (2 v)} {\ga \over v} \Big\{ \lim_{|x| \to \infty } {\log |x| \over 1 + \log (1 + x^2 )} {1 + \log (1 + x^2 ) \over 1 + \log (1 + x^2 v) } \Big\} ^{1 + \ga } = {\ga / 2^{1 + \ga } \over \sqrt{2 \pi }} v^{- 3 / 2} e^{- 1 / (2 v)} \non 
\end{align}
as $|x| \to \infty $ for each $v \in (0, \infty )$. 
Thus, by the dominated convergence theorem, we obtain 
\begin{align}
\lim_{|x| \to \infty } {f_{\rm EH} (x) \over |x|^{-1} (\log|x|)^{-1-\gamma}}
&= \int_{0}^{\infty} {\ga / 2^{1 + \ga } \over \sqrt{2 \pi }} v^{- 3 / 2} e^{- 1 / (2 v)} dv = {\ga \over 2^{1 + \ga }} \text{.} \non 
\end{align}
This complete the proof.

\section{Proof of Theorem 2.1}
\label{sec:proof_thm_normalized}

In this section, we prove Theorem 2.1. 
For $z \in \mathbb{R}$, we let 
\begin{align}
&f_0 (z) = {\rm{N}} (z | 0, 1) \text{,} \non \\
&f_1 (z) = \int_{0}^{\infty } {\rm{N}} (z | 0, u) H(u; \ga ) du \text{,} \quad \text{and} \non \\
&f(z) = (1 - s) f_0 (z) + s f_1 (z) \text{.} \non 
\end{align}

\begin{proof}[Proof of Theorem 2.1]
By Lemma \ref{lem:asymptotic} and 
part (iii) of Lemma \ref{lem:f0f1f}, we have for any $( \be , \si ) \in \mathbb{R} ^p \times (0, \infty )$ and any $i \in \mathcal{L}$, 
\begin{align}
{f( \ep _i ) / \si \over f( y_i )} &= {f_1 ( \ep _i ) / \si \over f_1 ( y_i )} {(1 - s) f_0 ( \ep _i ) / f_1 ( \ep _i ) + s \over (1 - s) f_0 ( y_i ) / f_1 ( y_i ) + s} \to 1 \non 
\end{align}
as $\om \to \infty $, where we write $\ep _i = ( y_i - x_{i}^{t} \be ) / \si $. 
Therefore, for any $( \be , \si ) \in \mathbb{R} ^p \times (0, \infty )$, 
\begin{align}
{p( \be , \si | \mathcal{D} ) \over p( \be , \si | \mathcal{D} ^{*} )} &= {p( \mathcal{D} ^{*} ) \prod_{i \in \mathcal{L}} f( y_i ) \over p( \mathcal{D} )} \prod_{i \in \mathcal{L}} {f( \ep _i ) / \si \over f( y_i )} \sim p( \mathcal{D} ^{*} ) / {p( \mathcal{D} ) \over \prod_{i \in \mathcal{L}} f( y_i )} \non 
\end{align}
as $\om \to \infty $. 
Now 
\begin{align}
{p( \mathcal{D} ) \over \prod_{i \in \mathcal{L}} f( y_i )} &= \int_{\mathbb{R} ^p \times (0, \infty )} \pi _{\si } ( \si ) \Big\{ \prod_{k = 1}^{p} {1 \over \si } \pi \Big( {\be _k \over \si } \Big) \Big\} \Big\{ \prod_{i \in \mathcal{K}} {f( \ep _i ) \over \si } \Big\} \Big\{ \prod_{i \in \mathcal{L}} {f( \ep _i ) / \si \over f( y_i )} \Big\} d( \be , \si ) \text{.} \non 
\end{align}
Then, by Lemma \ref{lem:DCT} below and by the dominated convergence theorem, 
\begin{align}
\lim_{\om \to \infty } {p( \mathcal{D} ) \over \prod_{i \in \mathcal{L}} f( y_i )} &= \int_{\mathbb{R} ^p \times (0, \infty )} \pi _{\si } ( \si ) \Big\{ \prod_{k = 1}^{p} {1 \over \si } \pi \Big( {\be _k \over \si } \Big) \Big\} \Big\{ \prod_{i \in \mathcal{K}} {f( \ep _i ) \over \si } \Big\} d( \be , \si ) = p( \mathcal{D} ^{*} ) \non 
\end{align}
and the result follows. 
\end{proof}

\begin{lem}
\label{lem:DCT} 
Under the assumptions of Theorem 2.1, there exists an integrable function $\overline{h} ( \be , \si )$ of $( \be , \si )$ which does not depend on $\om $ such that 
\begin{align}
\pi _{\si } ( \si ) \Big\{ \prod_{k = 1}^{p} {1 \over \si } \pi \Big( {\be _k \over \si } \Big) \Big\} \Big\{ \prod_{i \in \mathcal{K}} {f( \ep _i ) \over \si } \Big\} \prod_{i \in \mathcal{L}} {f( \ep _i ) / \si \over f( y_i )} \le \overline{h} ( \be , \si ) \non 
\end{align}
for all $( \be , \si ) \in \mathbb{R} ^p \times (0, \infty )$ for sufficiently large $\om $, where $\ep _i = ( y_i - x_{i}^{t} \be ) / \si $ for $i = 1, \dots , n$. 
\end{lem}

\begin{proof}
Let $\ep > 0$ be such that 
\begin{align}
\ep < {| b_i | \over 4 ( | x_{i, 1} | + \dots + | x_{i, p} | )} \non 
\end{align}
for all $i \in \mathcal{L}$. 
For $( \be , \si ) \in \mathbb{R} ^p \times (0, \infty )$ and $\om $, let 
\begin{align}
h_1 ( \be , \si ; \om ) &= \pi _{\si } ( \si ) \Big\{ \prod_{k = 1}^{p} {1 \over \si } \pi \Big( {\be _k \over \si } \Big) \Big\} \Big\{ \prod_{i \in \mathcal{K}} {f( \ep _i ) \over \si } \Big\} \Big\{ \prod_{i \in \mathcal{L}} {f( \ep _i ) / \si \over f( y_i )} \Big\} \chi _{[ - \ep \om , \ep \om ]^p} ( \be ) \text{,} \non 
\end{align}
where $\chi _{[ - \ep \om , \ep \om ]^p} ( \be ) = 1$ if $\be \in [ - \ep \om , \ep \om ]^p$ and $= 0$ otherwise, and similarly let 
\begin{align}
h_{2, k_0} ( \be , \si ; \om ) &= \pi _{\si } ( \si ) \Big\{ \prod_{k = 1}^{p} {1 \over \si } \pi \Big( {\be _k \over \si } \Big) \Big\} \Big\{ \prod_{i \in \mathcal{K}} {f( \ep _i ) \over \si } \Big\} \Big\{ \prod_{i \in \mathcal{L}} {f( \ep _i ) / \si \over f( y_i )} \Big\} \chi _{\mathbb{R} \setminus [ - \ep \om , \ep \om ]} ( \be _{k_0} ) \non 
\end{align}
for $k_0 = 1, \dots , p$. 
Then 
\begin{align}
\pi _{\si } ( \si ) \Big\{ \prod_{k = 1}^{p} {1 \over \si } \pi \Big( {\be _k \over \si } \Big) \Big\} \Big\{ \prod_{i \in \mathcal{K}} {f( \ep _i ) \over \si } \Big\} \prod_{i \in \mathcal{L}} {f( \ep _i ) / \si \over f( y_i )} \le h_1 ( \be , \si ; \om ) + \sum_{k_0 = 1}^{p} h_{2, k_0} ( \be , \si ; \om ) \label{lDCTp1} 
\end{align}
for all $( \be , \si ) \in \mathbb{R} ^p \times (0, \infty )$ and all $\om $.

First, we consider the first term in (\ref{lDCTp1}). 
For all $\be \in \mathbb{R} ^p$ and $\om $ satisfying $\be \in [ - \ep \om , \ep \om ]^p$ and all $i \in \mathcal{L}$, we have 
\begin{align}
\si | \ep _i | = | y_i - x_{i}^{t} \be | &\ge {| y_i | \over 2} + {| b_i | \om - | a_i | \over 2} - \sum_{k = 1}^{p} | x_{i, k} | | \be _k | \ge {| y_i | \over 2} + {| b_i | \om \over 4} - \sum_{k = 1}^{p} | x_{i, k} | \ep \om \ge {| y_i | \over 2} \text{.} \non 
\end{align}
Therefore, by parts (i) and (ii) of Lemma \ref{lem:f0f1f}, 
\begin{align}
h_1 ( \be , \si ; \om ) &\le \pi _{\si } ( \si ) \Big\{ \prod_{k = 1}^{p} {1 \over \si } \pi \Big( {\be _k \over \si } \Big) \Big\} \Big\{ {f(0) \over \si } \Big\} ^{| \mathcal{K} |} \Big\{ \prod_{i \in \mathcal{L}} {f(| y_i | / (2 \si )) / \si \over f(| y_i |)} \Big\} \chi _{[ - \ep \om , \ep \om ]^p} ( \be ) \non 
\end{align}
for all $( \be , \si ) \in \mathbb{R} ^p \times (0, \infty )$ and $\om $. 
Furthermore, by parts (iii), (iv), and (v) of Lemma \ref{lem:f0f1f}, 
\begin{align}
\prod_{i \in \mathcal{L}} {f(| y_i | / (2 \si )) / \si \over f(| y_i |)} &\le \prod_{i \in \mathcal{L}} {f_1 (| y_i | / (2 \si )) [(1 - s) \sup_{z \in \mathbb{R}} \{ f_0 (z) / f_1 (z) \} + s] / \si \over s f_1 (| y_i |)} \non \\
&\le \prod_{i \in \mathcal{L}} \Big( {2 \over s} {C_1 \over C_2} \Big[ {1 + \log (1 + | y_i |^2 ) \over 1 + \log \{ 1 + | y_i |^2 / (2 \si )^2 \} } \Big] ^{1 + \ga } \Big\{ (1 - s) \sup_{z \in \mathbb{R}} {f_0 (z) \over f_1 (z)} + s \Big\} \Big) \non \\
&\le M_1 (1 + \si ) < \infty \non 
\end{align}
for all $\si \in (0, \infty )$ and $\om $ for some $C_1 , C_2 , M_1 > 0$, where the last inequality follows 
since 
\begin{align}
{1 + \log (1 + r) \over 1 + \log \{ 1 + r / (2 \si )^2 \} } &\le {1 + \log \{ 1 + r / (2 \si )^2 \} + \log \{ 1 + (2 \si )^2 \} \over 1 + \log \{ 1 + r / (2 \si )^2 \} } \le 1 + \log \{ 1 + (2 \si )^2 \} \non 
\end{align}
for all $r \ge 0$. 
Thus, 
\begin{align}
h_1 ( \be , \si ; \om ) &\le \pi _{\si } ( \si ) \Big\{ \prod_{k = 1}^{p} {1 \over \si } \pi \Big( {\be _k \over \si } \Big) \Big\} \Big\{ {f(0) \over \si } \Big\} ^{| \mathcal{K} |} M_1 (1 + \si ) \non 
\end{align}
for all $( \be , \si ) \in \mathbb{R} ^p \times (0, \infty )$ and $\om $, which is an integrable function of $( \be , \si )$ since by assumption the prior mean of $\si ^{- | \mathcal{K} |}$ is finite and $| \mathcal{K} | \ge p \ge 1$. 

Next, we consider the second term in (\ref{lDCTp1}). 
Fix $k_0 = 1, \dots , p$. 
Let $i_0 = \min \mathcal{K}$. 
Let, for $( \be , \si ) \in \mathbb{R} ^p \times (0, \infty )$ and $\om $, 
\begin{align}
h_{2, k_0 , 1} ( \be , \si ) &= \pi _{\si } ( \si ) \Big\{ \prod_{k \in \{ 1, \dots , p \} \setminus \{ k_0 \} } {1 \over \si } \pi \Big( {\be _k \over \si } \Big) \Big\} {| x_{i_0 , k_0} | \over \si } f( \ep _{i_0} ) \non 
\end{align}
and 
\begin{align}
h_{2, k_0 , 2} ( \be , \si ; \om ) &= {1 \over \si } \pi \Big( {\be _{k_0} \over \si } \Big) {1 \over | x_{i_0 , k_0} |} \Big\{ \prod_{i \in \mathcal{K} \setminus \{ i_0 \} } {f( \ep _i ) \over \si } \Big\} \Big\{ \prod_{i \in \mathcal{L}} {f( \ep _i ) / \si \over f( y_i )} \Big\} \chi _{\mathbb{R} \setminus [ - \ep \om , \ep \om ]} ( \be _{k_0} ) \text{.} \non 
\end{align}
Then for all $( \be , \si ) \in \mathbb{R} ^p \times (0, \infty )$ and $\om $, 
\begin{align}
h_{2, k_0} ( \be , \si ; \om ) &= h_{2, k_0 , 1} ( \be , \si ) h_{2, k_0 , 2} ( \be , \si ; \om ) \text{.} \label{lDCTp3} 
\end{align}
We have that 
\begin{align}
\int_{\mathbb{R} ^p} h_{2, k_0 , 1} ( \be , \si ) d\be = \int_{\mathbb{R} ^{p - 1}} \Big\{ \int_{0}^{\infty } h_{2, k_0 , 1} ( \be , \si ) d{\be _{k_0}} \Big\} d( \be \setminus \be _{k_0} ) = \pi _{\si } ( \si ) \label{lDCTp4} 
\end{align}
for all $\si \in (0, \infty )$. 
On the other hand, by assumption (A.2) and by parts (iii) and (v) of Lemma \ref{lem:f0f1f}, 
\begin{align}
h_{2, k_0 , 2} ( \be , \si ; \om ) &\le {\sup_{\th \in \mathbb{R}} \{ | \th |^c \pi ( \th ) \} \over \si ^{1 - c} | \be _{k_0} |^c | x_{i_0 , k_0} |} \Big\{ \prod_{i \in \mathcal{K} \setminus \{ i_0 \} } {f_1 ( \ep _i ) \over \si } \Big\} \Big\{ \prod_{i \in \mathcal{L}} {f_1 ( \ep _i ) / \si \over s f_1 ( y_i )} \Big\} \non \\
&\quad \times \Big\{ (1 - s) \sup_{z \in \mathbb{R}} {f_0 (z) \over f_1 (z)} + s \Big\} ^{n - 1} \chi _{\mathbb{R} \setminus [ - \ep \om , \ep \om ]} ( \be _{k_0} ) \non \\
&\le {\sup_{\th \in \mathbb{R}} \{ | \th |^c \pi ( \th ) \} \over \si ^{1 - c} \ep ^c \om ^c | x_{i_0 , k_0} |} \Big\{ (1 - s) \sup_{z \in \mathbb{R}} {f_0 (z) \over f_1 (z)} + s \Big\} ^{n - 1} \Big\{ \prod_{i \in \mathcal{K} \setminus \{ i_0 \} } {f_1 ( \ep _i ) \over \si } \Big\} \non \\
&\quad \times \Big[ \prod_{i \in \mathcal{L}} {| y_i | \{ 1 + \log (1 + | y_i |^2 ) \} ^{1 + \ga } \over {C_2}' s} \Big] \prod_{i \in \mathcal{L}} {f_1 ( \ep _i ) \over \si } \non \\
&\le {\sup_{\th \in \mathbb{R}} \{ | \th |^c \pi ( \th ) \} \over \si ^{1 - c} \ep ^c | x_{i_0 , k_0} |} \Big\{ (1 - s) \sup_{z \in \mathbb{R}} {f_0 (z) \over f_1 (z)} + s \Big\} ^{n - 1} \non \\
&\quad \times \Big[ \sup_{\substack{\text{$\om $: sufficiently large} \\ \text{($\om \ge 1$, for example)}}} {\prod_{i \in \mathcal{L}} \{ 1 + \log (1 + | a_i + b_i \om |^2 ) \} ^{1 + \ga } \over ( {C_2}' s)^{| \mathcal{L} |} \om ^c} \Big] \non \\
&\quad \times \Big\{ \prod_{i \in \mathcal{K} \setminus \{ i_0 \} } {f_1 ( \ep _i ) \over \si } \Big\} \prod_{i \in \mathcal{L}} {| y_i | f_1 ( \ep _i ) \over \si } \non \\
&\le {M_2 \over \si ^{1 - c}} \Big( \prod_{i \in \mathcal{K} \setminus \{ i_0 \} } {1 \over \si + | y_i - x_{i}^{t} \be |} \Big) \prod_{i \in \mathcal{L}} {\si + | x_{i}^{t} \be | \over \si } \non \\
&\le {M_2 \over \si ^{1 - c}} \Big( \max \Big\{ 1, {1 \over \si ^{n - 1}} \Big\} \Big) \Big( \prod_{i \in \mathcal{K} \setminus \{ i_0 \} } {1 \over 1 + | y_i - x_{i}^{t} \be |} \Big) \prod_{i \in \mathcal{L}} (1 + | x_{i}^{t} \be |) \non 
\end{align}
for all $( \be , \si ) \in \mathbb{R} ^p \times (0, \infty )$ and $\om $ for some ${C_2}' , M_2 > 0$, where the fourth inequality follows from parts (vi) and (vii) of Lemma \ref{lem:f0f1f}. 
Thus, by Lemma \ref{lem:linear} (applied to $\prod_{i \in \mathcal{K} \setminus \{ i_0 \} } \{ 1 / (1 + | y_i - x_{i}^{t} \be |) \} $) and by assumption (A.1), 
\begin{align}
&h_{2, k_0 , 2} ( \be , \si ; \om ) / \Big\{ {M_2 \over \si ^{1 - c}} \Big( \max \Big\{ 1, {1 \over \si ^{n - 1}} \Big\} \Big)\Big\} \non \\
&\le \sup_{\be \in \{ \bet \in \mathbb{R} ^p | | \bet | \le R \} } {\prod_{i \in \mathcal{L}} (1 + | x_{i}^{t} \be |) \over \prod_{i \in \mathcal{K} \setminus \{ i_0 \} } (1 + | y_i - x_{i}^{t} \be |)} + {\{ 1 + ( \max_{1 \le i \le n} | x_i | ) | \be | \} ^{| \mathcal{L} |} \over (1 + \de | \be | )^{| \mathcal{K} | - p}} \non \\
&\le \sup_{\be \in \{ \bet \in \mathbb{R} ^p | | \bet | \le R \} } {\prod_{i \in \mathcal{L}} (1 + | x_{i}^{t} \be |) \over \prod_{i \in \mathcal{K} \setminus \{ i_0 \} } (1 + | y_i - x_{i}^{t} \be |)} + \Big( 1 + {\max_{1 \le i \le n} | x_i | \over \de } \Big) ^{| \mathcal{L} |} \label{lDCTp5} 
\end{align}
for all $( \be , \si ) \in \mathbb{R} ^p \times (0, \infty )$ and $\om $ for some $R > 0$ and $\de > 0$. 
Hence, combining (\ref{lDCTp3}) and (\ref{lDCTp5}), we obtain 
\begin{align}
h_{2, k_0} ( \be , \si ; \om ) &\le h_{2, k_0 , 1} ( \be , \si ) {M_2 \over \si ^{1 - c}} \Big( \max \Big\{ 1, {1 \over \si ^{n - 1}} \Big\} \Big) \non \\
&\quad \times \Big\{ \sup_{\be \in \{ \bet \in \mathbb{R} ^p | | \bet | \le R \} } {\prod_{i \in \mathcal{L}} (1 + | x_{i}^{t} \be |) \over \prod_{i \in \mathcal{K} \setminus \{ i_0 \} } (1 + | y_i - x_{i}^{t} \be |)} + \Big( 1 + {\max_{1 \le i \le n} | x_i | \over \de } \Big) ^{| \mathcal{L} |} \Big\} \non 
\end{align}
for all $( \be , \si ) \in \mathbb{R} ^p \times (0, \infty )$ and $\om $, which is an integrable function of $( \be , \si )$ by (\ref{lDCTp4}) and assumption (A.3). 
\end{proof}

\section{Tail heaviness and posterior robustness}
We here consider robustness properties for the wider class of error distributions defined by replacing $H(u;\gamma)$ in Section~2.4 of the main text with $H(u; \ga , \de )$ given in the finite mixture (2.2). The density of $H(u;\ga,\delta)$, which is given in (2.4) of the main text, is shown below; 
\begin{equation}\label{EHE2}
H(u;\gamma, \delta)= C(\delta, \ga ) {1\over (1 + u)^{1 + \delta}} {1 \over \{ 1 + \log (1 + u) \} ^{1 + \ga }} \text{,} \quad u > 0 \text{.} \
\end{equation}
Note that the distribution in (\ref{EHE2}) reduces to $H(u;\gamma)$ used in the proposed distribution under $\delta=0$. The parameter $\delta$ is related to the decay of the density tail of (\ref{EHE2}), that is, $H(u;\gamma, b)\approx u^{-\delta-1}(\log u)^{-1-\gamma}$.
Hence, the tail gets heavier as $\delta$ decreases, and the EHE distribution, in fact, has the heaviest tail in this class of distributions. 
We show later in Theorem~\ref{thm:robust1} that, among the general class given by (\ref{EHE2}), only the proposed error distribution that is realized by setting $\delta = 0$ could attain the exact robustness property.

To discuss the posterior robustness, we target the {\it unnormalized} posterior distribution of $(\beta,\sigma)$ given by 
\begin{align}
\pit _{\de } ( \be , \si | \mathcal{D} ) &= \pi( \be , \si ) \prod_{i = 1}^{n} \Big\{ {1 \over \si } f \Big( {y_i - x_i^t \be \over \si } \Big) \Big\} \text{,} \label{post} 
\end{align}
where $\pi ( \be , \si )$ is a prior density and where for $z \in \mathbb{R}$, $f (z) = (1 - s) f_0 (z) + s f_1 (z)$ and $f_0 (z) = {\rm{N}} (z | 0, 1)$ as in Section \ref{sec:proof_thm_normalized}, but now 
\begin{align}
&f_1 (z) = \int_{0}^{\infty } {\rm{N}} (z | 0, u) H(u; \ga , \de ) du \text{.} \non 
\end{align}
If $s=0$, the heavily-tailed component disappears and the model is obviously sensitive to outliers, hence suppose $s>0$ in the following. 
Next, we assume that each outlier goes to infinity at its own specific rate. 
More precisely, the observed values of responses are parametrized by $\omega$ as $y_i=y_i(\omega)$, and $|y_i(\om)| \to \infty$ as $\omega \to \infty$ for $i \in \mathcal{L}$ while $y_i(\om)$ is constant for $i \in \mathcal{K} = \{ 1, \dots , n \} \setminus \mathcal{L}$. 
The posterior robustness considered here is defined as the property that the unnormalized posterior conditional on $\mathcal{D}$ approaches that based on $\mathcal{D} ^{*}$ as $\om\to\infty$.

\begin{thm}
\label{thm:robust1} 
For any compact set $K \subset \mathbb{R} ^p \times (0, \infty )$, we have 
\begin{equation}\label{conv}
{\pit _{\de } ( \be , \si | \mathcal{D} ) \over \pit _{\de } ( \be , \si | \mathcal{D} ^{*} )} / \prod_{i \in \mathcal{L}} f( y_i ) \to \si ^{2 | \mathcal{L} | \de } \non 
\end{equation}
uniformly in $(\beta,\sigma)\in K$ as $\om \to \infty $. 
In particular, the unnormalized posterior is robust if and only if $\delta=0$. 
\end{thm}

We again note that the general error distribution with $\delta=0$ is exactly the proposed EHE distribution, so that the above theorem indicates that the desirable robustness property is achieved only under the proposed EHE distribution among the general class of error distributions with the mixing distribution in (\ref{EHE2}). 
The asymptotic ratio $\si ^{2 | \mathcal{L} | \de }$ is obtained for the $t$-distribution with $\de $ 
degrees of freedom. 
In other words, the posterior robustness cannot be attained by any finite mixture of $t$-distributions.

Theorem \ref{thm:robust1} shows the uniform convergence on any compact set of the unnormalized posterior density based on (\ref{EHE2}) with $\de = 0$ and all observations to the corresponding one based on non-outlying observations. 
In order to rigorously prove convergence in distribution, we have to justify an interchange of limit and integral concerning the normalizing constant for each model. 
The set of three assumptions (A.1)-(A.3) in Theorem~2.1 is an example that justifies such computation.

\begin{proof}[Proof of Theorem \ref{thm:robust1}]
The normalized ratio of $\pit _{\de } ( \be , \si | \mathcal{D} )$ to $\pit _{\de } ( \be , \si | \mathcal{D} ^{*} )$ is 
\begin{align}
{\pit _{\de } ( \be , \si | \mathcal{D} ) \over \pit _{\de } ( \be , \si | \mathcal{D} ^{*} )} / \prod_{i \in \mathcal{L}} f( y_i ) &= \prod_{i \in \mathcal{L}} {f(( y_i - x_i^t \beta ) / \si ) / \si \over f( y_i )} \text{.} \non 
\end{align}
It is sufficient to show that 
\begin{align}
{f(( y_i - x_i^t \beta ) / \si ) / \si \over f( y_i )} &\to \si ^{2 \de } \non 
\end{align}
uniformly in $( \beta , \si ) \in K$ as $\om \to \infty $ for every $i \in \mathcal{L}$. 
Fix $i \in \mathcal{L}$. 
Let $M = \sup_{( \beta , \si ) \in K} | x_i^t \beta | \in [0, \infty )$. 
Let $\underline{\si } = \inf_{( \beta , \si ) \in K} \si \in (0, \infty )$ and $\overline{\si } = \sup_{( \beta , \si ) \in K} \si \in (0, \infty )$. 
Assume without loss of generality that 
$\om $ is sufficiently large so that $| y_i | \ge 2 M + 1$. 

We first consider the case of $s = 1$. 
Then 
\begin{align}
{f(( y_i - x_i^t \beta ) / \si ) / \si \over f( y_i )} &= {f_1 (( y_i - x_i^t \beta ) / \si ) / \si \over f_1 ( y_i )} \non \\
&= {1 \over \si } \frac{ \int_{0}^{\infty } {\rm{N}} (( y_i - x_i^t \beta ) / \si | 0, u) H(u; \ga , \de ) du }{ \int_{0}^{\infty } {\rm{N}} ( y_i | 0, u) H(u; \ga , \de ) du } \non \\
&= {| y_i - x_i^t \beta | \over \si^2 |y_i|} \frac{ \int_{0}^{\infty } v^{- 1 / 2} e^{- 1 / (2 v)} H(( | y_i - x_i^t \beta |^2 / \si ^2 ) v | \ga , \de ) dv }{ \int_{0}^{\infty } v^{- 1 / 2} e^{- 1 / (2 v)} H( | y_i |^2 v | \ga , \de ) dv } \text{,} \non 
\end{align}
where the last equality follows by making the change of variables $u = (| y_i - x_i^t \beta | / \si )^2 v$ in the numerator and by making the change of variables $u = | y_i |^2 v$ in the denominator. 
Therefore, 
\begin{align}
\Big| {f(( y_i - x_i^t \beta ) / \si ) / \si \over f( y_i )} - \si ^{2 \de } \Big| &\le \overline{\si } ^{2 \de } \frac{ \int_{0}^{\infty } v^{- 1 / 2} e^{- 1 / (2 v)} H( | y_i |^2 v | \ga , \de ) G(v) dv }{ \int_{0}^{\infty } v^{- 1 / 2} e^{- 1 / (2 v)} H( | y_i |^2 v | \ga , \de ) dv } \text{,} \non 
\end{align}
where 
\begin{align}
G(v) &= G(v; \beta , \si , \ga , \de , y_i , x_i ) = \Big| {| y_i - x_i^t \beta | \over \si ^{2 (1 + \de )} |y_i|} {H(( | y_i - x_i^t \beta |^2 / \si ^2 ) v | \ga , \de ) \over H( | y_i |^2 v | \ga , \de )} - 1 \Big| \non \\
&= \Big| {| y_i - x_i^t \beta | \over | y_i |} \Big( {1 + | y_i |^2 v \over \si ^2 + | y_i - x_i^t \beta |^2 v} \Big) ^{1 + \de } \Big[ {1 + \log (1 + | y_i |^2 v ) \over 1 + \log \{ 1 + ( | y_i - x_i^t \beta |^2 / \si ^2 ) v \} } \Big] ^{1 + \ga } - 1 \Big| \non 
\end{align}
for $v > 0$. 
Note that 
\begin{align}
F_1 (v) &\le {| y_i - x_i^t \beta | \over | y_i |} \Big( {1 + | y_i |^2 v \over \si ^2 + | y_i - x_i^t \beta |^2 v} \Big) ^{1 + \de } \Big[ {1 + \log (1 + | y_i |^2 v ) \over 1 + \log \{ 1 + ( | y_i - x_i^t \beta |^2 / \si ^2 ) v \} } \Big] ^{1 + \ga } \le F_2 (v) \text{,} \non 
\end{align}
where 
\begin{align}
&F_1 (v) = {| y_i | - M \over | y_i |} \Big\{ {1 + | y_i |^2 v \over \overline{\si } ^2 + (| y_i | + M)^2 v} \Big\} ^{1 + \de } \Big( {1 + \log (1 + | y_i |^2 v ) \over 1 + \log [1 + \{ ( | y_i | + M)^2 / \underline{\si } ^2 \} v]} \Big) ^{1 + \ga } \text{,} \non \\
&F_2 (v) = {| y_i | + M \over | y_i |} \Big\{ {1 + | y_i |^2 v \over \underline{\si } ^2 + (| y_i | - M)^2 v} \Big\} ^{1 + \de } \Big( {1 + \log (1 + | y_i |^2 v ) \over 1 + \log [1 + \{ ( | y_i | - M)^2 / \overline{\si } ^2 \} v]} \Big) ^{1 + \ga } \text{.} \non 
\end{align}
Then 
\begin{align}
G(v) &\le | F_1 (v) - 1| + | F_2 (v) - 1| \text{.} \non 
\end{align}
Therefore, 
\begin{align}
&\Big| {f(( y_i - x_i^t \beta ) / \si ) / \si \over f( y_i )} - \si ^{2 \de } \Big| \non \\
&\le \overline{\si } ^{2 \de } \frac{ \int_{0}^{\infty } v^{- 1 / 2} e^{- 1 / (2 v)} \tilde{H} (v) \{ | F_1 (v) - 1| + | F_2 (v) - 1| \} dv }{ \int_{0}^{\infty } v^{- 1 / 2} e^{- 1 / (2 v)} \tilde{H} (v) dv } \text{,} \label{tposp1} 
\end{align}
where 
\begin{align}
\tilde{H} (v) &= {H( | y_i |^2 v | \ga , \de ) \over H( | y_i |^2 | \ga , \de )} \text{.} \non 
\end{align}
The right-hand side of (\ref{tposp1}) is independent of $( \beta , \si )$. 
We have that $\lim_{\om \to \infty } (| F_1 (v) - 1 | + | F_2 (v) - 1 |) = 0$ for each $v > 0$ and that for $| y_i | \ge 1$, 
\begin{align}
v^{- 1 / 2} e^{- 1 / (2 v)} \tilde{H} (v) &= v^{- 1 / 2} \Big( {1 + | y_i |^2 \over 1 + | y_i |^2 v} \Big) ^{1 + \de } \Big\{ {1 + \log (1 + | y_i |^2 ) \over 1 + \log (1 + | y_i |^2 v)} \Big\} ^{1 + \ga } e^{- 1 / (2 v)} \non \\
&\begin{cases} \displaystyle 
\le 2^{1 + \de } v^{- 1 / 2 - 1 - \de } \max \{ 1, v^{- (1 + \ga )} \} e^{- 1 / (2 v)} \\ \displaystyle \to v^{- 1 / 2 - 1 - \de } e^{- 1 / (2 v)} \quad \text{as $\om \to \infty $} \end{cases} \non 
\end{align}
for all $v > 0$ by Lemma \ref{lem:IL}. 
Furthermore, 
\begin{align}
| F_1 (v) - 1| + | F_2 (v) - 1| &\le 2 + | F_1 (v) | + | F_2 (v) | \le 2 \{ 1 + F_2 (v) \} \non 
\end{align}
and, since $| y_i | \ge 2 M + 1 > M$, we have 
\begin{align}
F_2 (v) &= {| y_i | + M \over | y_i |} \Big\{ {1 + | y_i |^2 v \over \underline{\si } ^2 + (| y_i | - M)^2 v} \Big\} ^{1 + \de } \Big( {1 + \log (1 + | y_i |^2 v ) \over 1 + \log [1 + \{ ( | y_i | - M)^2 / \overline{\si } ^2 \} v]} \Big) ^{1 + \ga } \non \\
&\le 2 \Big\{ {1 \over \underline{\si } ^2} + {| y_i |^2 \over (| y_i | - M)^2} \Big\} ^{1 + \de } \Big( 1 + {\log {1 + | y_i |^2 v \over 1 + \{ ( | y_i | - M)^2 / \overline{\si } ^2 \} v]} \over 1 + \log [1 + \{ ( | y_i | - M)^2 / \overline{\si } ^2 \} v]} \Big) ^{1 + \ga } \non \\
&\le 2 \Big({1 \over \underline{\si } ^2} + 4 \Big) ^{1 + \de } \Big[ 1 + \Big| \log {1 + | y_i |^2 v \over 1 + \{ ( | y_i | - M)^2 / \overline{\si } ^2 \} v} \Big| \Big] ^{1 + \ga } \text{,} \non 
\end{align}
where 
\begin{align}
&\Big| \log {1 + | y_i |^2 v \over 1 + \{ ( | y_i | - M)^2 / \overline{\si } ^2 \} v} \Big| \non \\
&= \Big| \int_{( | y_i | - M)^2 / (| y_i | \overline{\si } )^2}^{1} {| y_i |^2 v \over 1 + | y_i |^2 v t} dt \Big| \le \int_{\min \{ 1, ( | y_i | - M)^2 / (| y_i | \overline{\si } )^2 \} }^{\max \{ 1, ( | y_i | - M)^2 / (| y_i | \overline{\si } )^2 \} } {1 \over t} dt \non \\
&\le {\max \{ 1, ( | y_i | - M)^2 / (| y_i | \overline{\si } )^2 \} - \min \{ 1, ( | y_i | - M)^2 / (| y_i | \overline{\si } )^2 \} \over \min \{ 1, ( | y_i | - M)^2 / (| y_i | \overline{\si } )^2 \} } \non \\
&= {| (| y_i | \overline{\si } )^2 - ( | y_i | - M)^2 | \over \min \{ (| y_i | \overline{\si } )^2 , ( | y_i | - M)^2 \} } \le {(| y_i | \overline{\si } )^2 \over ( | y_i | - M)^2} + {( | y_i | - M)^2 \over (| y_i | \overline{\si } )^2} \le (2 \overline{\si } )^2 + (1 / \overline{\si } )^2 \text{.} \non 
\end{align}
Thus, by the dominated convergence theorem, the right-hand side of (\ref{tposp1}) converges to zero as $\om \to \infty $. 

Next we consider the case of $s \in (0, 1)$. 
Then we have 
\begin{align}
{f(( y_i - x_i^t \beta ) / \si ) / \si \over f( y_i )} &= {f_1 (( y_i - x_i^t \beta ) / \si ) / \si \over f_1 ( y_i )} \frac{ \displaystyle s + (1 - s) {f_0 (( y_i - x_i^t \beta ) / \si ) \over f_1 (( y_i - x_i^t \beta ) / \si )} }{ \displaystyle s + (1 - s) {f_0 ( y_i ) \over f_1 ( y_i )} } \text{.} \non 
\end{align}
Therefore, 
\begin{align}
\Big| {f(( y_i - x_i^t \beta ) / \si ) / \si \over f( y_i )} - \si ^{2 \de } \Big| &\le \overline{\si } ^{2 \de } \Big| {f(( y_i - x_i^t \beta ) / \si ) / \si \over f( y_i ) \si ^{2 \de }} - 1 \Big| \non \\
&\le \overline{\si } ^{2 \de } \Big[ \Big\{ \Big| {f_1 (( y_i - x_i^t \beta ) / \si ) / \si \over f_1 ( y_i ) \si ^{2 \de }} - 1 \Big| + 1 \Big\} \non \\
&\quad \times \Big\{ \Big| \frac{ \displaystyle s + (1 - s) {f_0 (( y_i - x_i^t \beta ) / \si ) \over f_1 (( y_i - x_i^t \beta ) / \si )} }{ \displaystyle s + (1 - s) {f_0 ( y_i ) \over f_1 ( y_i )} } - 1 \Big| + 1 \Big\} - 1 \Big] \text{.} \non 
\end{align}
By the result for $s = 1$, 
\begin{align}
\sup_{( \beta , \si ) \in K} \Big| {f_1 (( y_i - x_i^t \beta ) / \si ) / \si \over f_1 ( y_i ) \si ^{2 \de }} - 1 \Big| &\le {1 \over \underline{\si } ^{2 \de }} \sup_{( \beta , \si ) \in K} \Big| {f_1 (( y_i - x_i^t \beta ) / \si ) / \si \over f_1 ( y_i )} - \si ^{2 \de } \Big| \to 0 \non 
\end{align}
as $\om \to \infty $. 
On the other hand, 
\begin{align}
\Big| \frac{ \displaystyle s + (1 - s) {f_0 (( y_i - x_i^t \beta ) / \si ) \over f_1 (( y_i - x_i^t \beta ) / \si )} }{ \displaystyle s + (1 - s) {f_0 ( y_i ) \over f_1 ( y_i )} } - 1 \Big| &\le \Big| \frac{ \displaystyle s }{ \displaystyle s + (1 - s) {f_0 ( y_i ) \over f_1 ( y_i )} } - 1 \Big| + {1 - s \over s} {f_0 (( y_i - x_i^t \beta ) / \si ) \over f_1 (( y_i - x_i^t \beta ) / \si )} \text{.} \label{tposp2} 
\end{align}
by Lemma \ref{lem:asymptotic}, the first term on the right side of (\ref{tposp2}) converges to zero as $\om \to \infty $. 
Since $f_0 (z) = f_0 (|z|)$ and $f_1 (z) = f_1 (|z|)$ are nonincreasing functions of $|z|$ and since $M \le | y_i | / 2 \le | y_i |$, it follows that 
\begin{align}
{f_0 (( y_i - x_i^t \beta ) / \si ) \over f_1 (( y_i - x_i^t \beta ) / \si )} &\le {f_0 ((| y_i | - M) / \overline{\si } ) \over f_1 ((| y_i | + M) / \underline{\si } )} = {f_0 ((| y_i | - M) / \overline{\si } ) \over f_1 ((| y_i | - M) / \overline{\si } )} {f_1 ((| y_i | - M) / \overline{\si } ) \over f_1 ((| y_i | + M) / \underline{\si } )} \non \\
&\le {f_0 ((| y_i | - M) / \overline{\si } ) \over f_1 ((| y_i | - M) / \overline{\si } )} {f_1 (| y_i | / (2 \overline{\si } )) \over f_1 (| y_i | / ( \underline{\si } / 2))} \text{,} \non 
\end{align}
where 
\begin{align}
\lim_{\om \to \infty } {f_0 ((| y_i | - M) / \overline{\si } ) \over f_1 ((| y_i | - M) / \overline{\si } )} = 0 \text{.} \non 
\end{align}
Furthermore, 
\begin{align}
{f_1 (| y_i | / (2 \overline{\si } )) \over f_1 (| y_i | / ( \underline{\si } / 2))} &= \frac{ \int_{0}^{\infty } {\rm{N}} (| y_i | / (2 \overline{\si } ) | 0, u) H(u; \ga , \de ) du }{ \int_{0}^{\infty } {\rm{N}} (| y_i | / ( \underline{\si } / 2) | 0, u) H(u; \ga , \de ) du } \non \\
&= {\underline{\si } \over 4 \overline{\si }} \frac{ \int_{0}^{\infty } {\rm{N}} (| y_i | | 0, v) H(v / (2 \overline{\si } )^2 ; \ga , \de ) dv }{ \int_{0}^{\infty } {\rm{N}} (| y_i | | 0, v) H(v / ( \underline{\si } / 2)^2 ; \ga , \de ) dv } \non \\
&\to \Big( {4 \overline{\si } \over \underline{\si }} \Big) ^{1 + 2 \de } \non 
\end{align}
as $\om \to \infty $ by Lemma \ref{lem:ratio} since 
\begin{align}
{H(v / (2 \overline{\si } )^2 ; \ga , \de ) \over H(v / ( \underline{\si } / 2)^2 ; \ga , \de )} &= \Big\{ {1 + v / ( \underline{\si } / 2)^2 \over 1 + v / (2 \overline{\si } )^2} \Big\} ^{1 + \de } \Big[ {1 + \log \{ 1 + v / ( \underline{\si } / 2)^2 \} \over 1 + \log \{ 1 + v / (2 \overline{\si } )^2 \} } \Big] ^{1 + \ga } \to \Big( {4 \overline{\si } \over \underline{\si }} \Big) ^{2 (1 + \de )} \non 
\end{align}
as $v \to \infty $ by Lemma \ref{lem:IL}. 
Thus, we conclude that 
\begin{align}
\sup_{( \beta , \si ) \in K} \Big| {f(( y_i - x_i^t \beta ) / \si ) / \si \over f( y_i )} - \si ^{2 \de } \Big| \to 0 \non 
\end{align}
as $\om \to \infty$. 
\end{proof}

\section{Posterior Moments of $\be $ and $\si ^2$}
Here we prove Proposition~2.2, the existence of posterior moments of $(\be,\sigma ^2)$. The proof is given for a slightly generalized model as given below.

Let $f(z)$, $z \in \mathbb{R}$, be a symmetric bounded error density. 
For each $k = 1, \dots , p$, let $\pi _k ( \th )$, $\th \in \mathbb{R}$, be a proper prior density and let $\nu _k \in \{ 0, 1 \} $. 
Let $a_{\si } , b_{\si } > 0$. 
Suppose that for $i = 1, \dots , n$ and $k = 1, \dots , p$, 
\begin{align}
&y_i \sim {1 \over \si } f \Big( {y_i - x_{i}^{t} \be \over \si } \Big) \text{,} \quad \be _k \sim {1 \over \si ^{\nu _k}} \pi _k \Big( {\be _k \over \si ^{\nu _k}} \Big) \text{,} \quad \si \sim 2 \si {\rm{IG}} ( \si ^2 | a_{\si } , b_{\si } ) \text{.} \non 
\end{align}

\begin{prp}
Let $k_0 = 1, \dots , p$. 
Suppose that $\sup_{\th \in \mathbb{R}} \{ | \th |^c \pi _{k_0} ( \th ) \} < \infty $ for $0 < c \le n$. 
Then $E[ | \be _{k_0} |^c | \mathcal{D} ] < \infty $. 
\end{prp}

\begin{proof}
We have 
\begin{align}
&p( \mathcal{D} ) E[ | \be _{k_0} |^c | \mathcal{D} ] \non \\
&= \int_{\mathbb{R} ^p \times (0, \infty )} \Big( 2 \si {\rm{IG}} ( \si ^2 | a_{\si } , b_{\si } ) | \be _{k_0} |^c {1 \over \si ^{\nu _{k_0}}} \pi _{k_0} \Big( {\be _{k_0} \over \si ^{\nu _{k_0}}} \Big) \Big[ \prod_{\substack{1 \le k \le p \\ k \neq k_0}} \Big\{ {1 \over \si ^{\nu _k}} \pi _k \Big( {\be _k \over \si ^{\nu _k}} \Big) \Big\} \Big] \non \\
&\quad \times {| x_{1, k_0} | \over \si } f \Big( {y_1 - x_{1}^{t} \be \over \si } \Big) {1 \over | x_{1, k_0} |} \Big[ \prod_{i = 2}^{n} \Big\{ {1 \over \si } f \Big( {y_i - x_{i}^{t} \be \over \si } \Big) \Big\} \Big] \Big) d( \be , \si ) \non \\
&\le \int_{\mathbb{R} ^p \times (0, \infty )} \Big( 2 \si {\rm{IG}} ( \si ^2 | a_{\si } , b_{\si } ) \Big[ \prod_{\substack{1 \le k \le p \\ k \neq k_0}} \Big\{ {1 \over \si ^{\nu _k}} \pi _k \Big( {\be _k \over \si ^{\nu _k}} \Big) \Big\} \Big] {| x_{1, k_0} | \over \si } f \Big( {y_1 - x_{1}^{t} \be \over \si } \Big) \non \\
&\quad \times \si ^{\nu _{k_0} (c - 1)} [ \sup_{\th \in \mathbb{R}} \{ | \th |^c \pi _{k_0} ( \th ) \} ] {1 \over | x_{1, k_0} |} \Big\{ {\sup_{z \in \mathbb{R}} f(z) \over \si } \Big\} ^{n - 1} \Big) d( \be , \si ) \non \\
&= \int_{0}^{\infty } 2 \si {\rm{IG}} ( \si ^2 | a_{\si } , b_{\si } ) \si ^{\nu _{k_0} (c - 1)} [ \sup_{\th \in \mathbb{R}} \{ | \th |^c \pi _{k_0} ( \th ) \} ] {1 \over | x_{1, k_0} |} {\{ \sup_{z \in \mathbb{R}} f(z) \} ^{n - 1} \over \si ^{n - 1}} d\si \text{,} \non 
\end{align}
which is finite since $\nu _{k_0} (c - 1) \le n - 1$ by assumption. 
\end{proof}

\begin{prp}
Suppose that $d \le n$. 
Then $E[ \si ^d | \mathcal{D} ] < \infty $. 
\end{prp}

\begin{proof}
We have 
\begin{align}
&p( \mathcal{D} ) E[ \si ^d | \mathcal{D} ] \non \\
&= \int_{\mathbb{R} ^p \times (0, \infty )} \si ^d 2 \si {\rm{IG}} ( \si ^2 | a_{\si } , b_{\si } ) \Big[ \prod_{k = 1}^{p} \Big\{ {1 \over \si ^{\nu _k}} \pi _k \Big( {\be _k \over \si ^{\nu _k}} \Big) \Big\} \Big] \Big[ \prod_{i = 1}^{n} \Big\{ {1 \over \si } f \Big( {y_i - x_{i}^{t} \be \over \si } \Big) \Big\} \Big] d( \be , \si ) \non \\
&\le \int_{\mathbb{R} ^p \times (0, \infty )} \si ^d 2 \si {\rm{IG}} ( \si ^2 | a_{\si } , b_{\si } ) \Big[ \prod_{k = 1}^{p} \Big\{ {1 \over \si ^{\nu _k}} \pi _k \Big( {\be _k \over \si ^{\nu _k}} \Big) \Big\} \Big] {\{ \sup_{z \in \mathbb{R}} f(z) \} ^n \over \si ^n} d( \be , \si ) \non \\
&= \int_{0}^{\infty } \si ^d 2 \si {\rm{IG}} ( \si ^2 | a_{\si } , b_{\si } ) {\{ \sup_{z \in \mathbb{R}} f(z) \} ^n \over \si ^n} d\si \text{,} \non 
\end{align}
which is finite by assumption. 
\end{proof}

\section{Additional experiment in simulation study}

\subsection{Sensitivity analysis}

To evaluate the effect of hyperparameters on the posterior inference, we repeated the posterior analysis with different choice of hyperparameters. For the shape parameter of $H$-distribution, we additionally considered $\ga = 0.5$ and $\ga = 0.2$, in addition to our choice in the main text, $\ga = 1$. For the degree-of-freedom parameter of the single $t$-distribution and the finite mixture, we considered $\nu = 2.1$ as the ``heaviest'' $t$-distribution with finite mean and variance. The result of posterior analysis is reported in Table~\ref{tab:sim-p20-add1} and \ref{tab:sim-p20-add2} in the same style of Table~1 in the main text. 
It is observed that the EH methods with two different values of $\gamma$ perform almost in the same way as the EH method with $\gamma=1$.

\subsection{Regression with less predictors}

The LPTN models are estimated by the random-walk Metropolis-Hastings algorithm, which requires many iterations in posterior sampling for convergence. While keeping the fairness in the number of iterations, we conduct another experiment that favors the LPTN models by partly eliminating the convergence issue in the LPTN models. The additional simulation study is based on the same settings in Section 4, except that the number of predictors is now $p=10$.

The results are summarized in Tables~\ref{tab:sim-ap1} and \ref{tab:sim-ap2}. The IFs of the LPTN models are improved, but still significantly higher than the others. The LPTN model with $\rho = 0.9$ improves the accuracy of point and interval estimations and is now competitive with the proposed models, while the other LPTN model with $\rho = 0.7$ still provides interval estimates with lower coverage probabilities. This result illustrates the difficulty in tuning the hyperparameters in the class of LPTN distributions, which contrasts the proposed model with no hyperparameter that is sensitive to the posterior result.

\subsection{Computational time with large sample size}

We also measured the actual computation time of the five methods (EH, LP1, T3, MT and N) under different sample sizes. 
We considered four scenarios of $n$, that is, $n\in\{300, 1200, 2100, 3000\}$.
For each $n$, synthetic data is generated using the model with $(100\omega,\mu)=(5, 10)$, and 3000 posterior samples are generated for each method. 
To assess computation time that takes account of sampling efficiency, we compute ${\rm CPT}\times{\rm IF}$, where CPT is the actual computation time to generate 3000 posterior samples and IF is the inefficiency factor.
Note that this quantity can be regarded as computation time to generate 3000 independent posterior samples. 
The experiment was performed on a PC with 3.2 GHz 8-Core Intel Xeon W 8 Core Processor with approximately 32GB RAM.
The results are reported in Table \ref{tab:sim-cpt}.
It is observed that the EH and LP1 methods take more computation time than the others, which would be reasonable price to pay for their posterior robustness.
Comparing EH and LP1, EH is computationally more efficient than LP1.

\begin{table}[htp!]
\caption{\footnotesize Average values of RMSE and IF of the proposed extremely-heavy tailed (EH) distribution with $\gamma=0.5$ and $\gamma=0.2$, and its adaptive version (aEH) with three different priors for $\gamma$, $t$-distribution (T) with $\nu=2.1$ degrees of freedom and two component mixture of normal and $t$-distribution (MT) with $\nu=2.1$ degrees of freedom, based on 500 replications in 9 combinations of $(100\omega, \mu)$ with $p=20$.
All values are multiplied by 100.
\label{tab:sim-p20-add1}}
\begin{center}
\begin{tabular}{ccccccccccccc}
\hline
& & EH & EH & aEH & aEH & aEH &T & MT\\
 & $(100\omega, \mu)$  & {\footnotesize $\gamma=0.5$} & {\footnotesize $\gamma=0.2$} & {\footnotesize ${\rm Ga}(10,100)$} & {\footnotesize ${\rm Ga}(1,1)$} & {\footnotesize ${\rm Ga}(10,10)$} &
 {\footnotesize $\nu=2.1$} & {\footnotesize $\nu=2.1$} \\
 \hline
 & (0, --) & 6.32 & 6.33 & 6.32 & 6.34 & 6.33 & 7.03 & 6.33 \\
 & (5, 5) & 6.99 & 7.23 & 6.94 & 7.47 & 7.14 & 7.25 & 6.99 \\
 & (10, 5) & 10.79 & 12.41 & 9.58 & 8.64 & 8.55 & 8.03 & 7.96 \\
 & (5, 10) & 6.54 & 6.53 & 6.56 & 6.80 & 6.74 & 7.08 & 6.78 \\
RMSE & (10, 10) & 6.85 & 6.81 & 6.91 & 7.57 & 7.44 & 7.39 & 7.30 \\
 & (5, 15) & 6.54 & 6.52 & 6.56 & 6.76 & 6.73 & 7.08 & 6.80 \\
 & (10, 15) & 6.87 & 6.81 & 6.92 & 7.36 & 7.30 & 7.28 & 7.19 \\
 & (5, 20) & 6.48 & 6.46 & 6.49 & 6.67 & 6.64 & 7.02 & 6.72 \\
 & (10, 20) & 6.84 & 6.79 & 6.89 & 7.23 & 7.20 & 7.21 & 7.12 \\
\hline
 & (0, --) & 0.98 & 0.98 & 0.99 & 1.43 & 1.07 & 2.61 & 1.07 \\
 & (5, 5) & 1.74 & 1.89 & 1.79 & 4.84 & 3.84 & 2.42 & 1.99 \\
 & (10, 5) & 2.75 & 2.63 & 2.98 & 5.30 & 6.10 & 2.26 & 2.11 \\
 & (5, 10) & 1.42 & 1.26 & 1.54 & 3.45 & 3.06 & 2.35 & 1.92 \\
IF & (10, 10) & 1.87 & 1.54 & 2.23 & 5.29 & 4.87 & 2.11 & 1.97 \\
 & (5, 15) & 1.40 & 1.25 & 1.53 & 3.09 & 2.81 & 2.33 & 1.90 \\
 & (10, 15) & 1.86 & 1.54 & 2.19 & 4.55 & 4.29 & 2.09 & 1.95 \\
 & (5, 20) & 1.40 & 1.24 & 1.54 & 2.90 & 2.70 & 2.34 & 1.91 \\
 & (10, 20) & 1.86 & 1.55 & 2.17 & 4.14 & 4.00 & 2.08 & 1.93 \\

 \hline
\end{tabular}
\end{center}
\end{table}

\begin{table}[htp!]
\caption{\footnotesize Average values of CP and AL of the proposed extremely-heavy tailed (EH) distribution with $\gamma=0.5$ and $\gamma=2$, and its adaptive version (aEH) with three different priors for $\gamma$, $t$-distribution (T) with $\nu=2.1$ degrees of freedom and two component mixture of normal and $t$-distribution (MT) with $\nu=2.1$ degrees of freedom, based on 500 replications in 9 combinations of $(100\omega, \mu)$ with $p=20$.
All values are multiplied by 100.
\label{tab:sim-p20-add2}}
\begin{center}
\begin{tabular}{ccccccccccccc}
\hline
& & EH & EH & aEH & aEH & aEH &T & MT\\
 & $(100\omega, \mu)$  & {\footnotesize $\gamma=0.5$} & {\footnotesize $\gamma=0.2$} & {\footnotesize ${\rm Ga}(10,100)$} & {\footnotesize ${\rm Ga}(1,1)$} & {\footnotesize ${\rm Ga}(10,10)$} &
 {\footnotesize $\nu=2.1$} & {\footnotesize $\nu=2.1$} \\
 \hline
 & (0, --) & 94.8 & 94.8 & 94.9 & 94.9 & 94.9 & 92.2 & 94.8 \\
 & (5, 5) & 94.8 & 94.4 & 94.7 & 92.8 & 94.0 & 93.6 & 94.6 \\
 & (10, 5) & 93.1 & 92.2 & 93.4 & 92.1 & 91.9 & 93.7 & 94.0 \\
 & (5, 10) & 95.0 & 94.9 & 95.0 & 94.4 & 94.6 & 94.4 & 95.3 \\
CP & (10, 10) & 94.8 & 94.8 & 94.7 & 93.1 & 93.4 & 95.7 & 95.9 \\
 & (5, 15) & 95.1 & 95.0 & 94.9 & 94.2 & 94.3 & 94.1 & 95.1 \\
 & (10, 15) & 94.5 & 94.6 & 94.4 & 93.4 & 93.4 & 95.8 & 96.0 \\
 & (5, 20) & 95.0 & 95.0 & 95.0 & 94.6 & 94.7 & 94.7 & 95.6 \\
 & (10, 20) & 94.7 & 94.6 & 94.6 & 94.0 & 93.9 & 96.2 & 96.5 \\
 \hline
 & (0, --) & 24.7 & 24.7 & 24.7 & 24.6 & 24.7 & 24.7 & 24.7 \\
 & (5, 5) & 27.0 & 27.6 & 26.9 & 27.1 & 26.8 & 27.0 & 27.0 \\
 & (10, 5) & 33.6 & 36.7 & 31.6 & 30.6 & 29.9 & 30.3 & 30.3 \\
 & (5, 10) & 25.8 & 25.7 & 25.8 & 26.1 & 26.0 & 27.1 & 27.1 \\
AL & (10, 10) & 26.9 & 26.7 & 27.1 & 27.7 & 27.6 & 30.3 & 30.2 \\
 & (5, 15) & 25.6 & 25.5 & 25.7 & 25.8 & 25.8 & 27.0 & 27.0 \\
 & (10, 15) & 26.7 & 26.5 & 26.8 & 27.2 & 27.1 & 30.1 & 30.1 \\
 & (5, 20) & 25.7 & 25.6 & 25.7 & 25.9 & 25.9 & 27.2 & 27.2 \\
 & (10, 20) & 26.7 & 26.5 & 26.7 & 27.0 & 27.0 & 30.3 & 30.3 \\
 & (0, --) & 0.98 & 0.98 & 0.99 & 1.43 & 1.07 & 2.61 & 1.07 \\
\hline
\end{tabular}
\end{center}
\end{table}

\begin{table}[htp!]
\caption{\footnotesize Average values of RMSE and IF of the proposed extremely-heavy tailed distribution with fixed $\gamma$ (EH) and estimated gamma (aEH), log-Pareto normal distribution with $\rho=0.9$ (LP1) and $\rho=0.7$ (LP2), Cauchy distribution (C), $t$-distribution with 3 degrees of freedom (T3) and estimated degrees of freedom (T), based on 500 replications in 9 combinations of $(100\omega, \mu)$ with $p=10$.
All values are multiplied by 100.
\label{tab:sim-ap1}}
\begin{center}
\begin{tabular}{ccccccccccccc}
\hline
 & $(100\omega, \mu)$ & & EH & aEH & LP1 & LP2 & C & T3 & T & MT & N \\
 \hline
  & (0, --) &  & 6.18 & 6.18 & 6.41 & 7.66 & 7.71 & 6.64 & 6.42 & 6.19 & 6.18 \\
 & (5, 5) &  & 6.68 & 6.72 & 6.85 & 8.07 & 7.76 & 7.01 & 7.39 & 6.60 & 11.78 \\
 & (10, 5) &  & 8.14 & 8.09 & 8.42 & 8.67 & 8.14 & 8.28 & 10.16 & 8.80 & 18.73 \\
 & (5, 10) &  & 6.39 & 6.44 & 6.48 & 7.82 & 7.73 & 6.72 & 7.09 & 6.34 & 21.12 \\
RMSE & (10, 10) &  & 6.82 & 6.95 & 6.80 & 8.01 & 7.76 & 7.21 & 10.28 & 8.11 & 35.68 \\
 & (5, 15) &  & 6.44 & 6.47 & 6.55 & 7.80 & 7.72 & 6.69 & 6.93 & 6.40 & 30.92 \\
 & (10, 15) &  & 6.87 & 6.95 & 6.75 & 7.99 & 7.81 & 7.02 & 10.65 & 7.41 & 53.56 \\
 & (5, 20) &  & 6.37 & 6.40 & 6.46 & 7.72 & 7.72 & 6.61 & 6.74 & 6.33 & 40.57 \\
 & (10, 20) &  & 6.76 & 6.85 & 6.69 & 8.02 & 7.71 & 6.83 & 10.58 & 11.06 & 70.79 \\
\hline
 & (0, --) &  & 1.02 & 1.02 & 27.99 & 41.03 & 4.32 & 2.09 & 1.84 & 0.99 & 0.98 \\
 & (5, 5) &  & 2.25 & 2.67 & 27.42 & 39.60 & 4.05 & 1.95 & 1.83 & 1.33 & 0.98 \\
 & (10, 5) &  & 3.72 & 4.63 & 27.63 & 38.83 & 3.79 & 1.85 & 1.89 & 2.05 & 0.98 \\
 & (5, 10) &  & 2.16 & 2.49 & 27.59 & 40.12 & 4.00 & 1.90 & 1.81 & 1.28 & 0.98 \\
IF & (10, 10) &  & 3.43 & 4.10 & 27.25 & 39.16 & 3.71 & 1.72 & 2.03 & 1.54 & 0.98 \\
 & (5, 15) &  & 2.17 & 2.46 & 27.63 & 40.04 & 4.04 & 1.88 & 1.81 & 1.28 & 0.98 \\
 & (10, 15) &  & 3.45 & 4.00 & 27.37 & 39.40 & 3.69 & 1.67 & 2.14 & 1.59 & 0.98 \\
 & (5, 20) &  & 2.16 & 2.41 & 27.73 & 40.11 & 4.04 & 1.87 & 1.80 & 1.26 & 0.98 \\
 & (10, 20) &  & 3.45 & 3.89 & 27.41 & 39.63 & 3.66 & 1.66 & 2.22 & 1.61 & 0.98 \\
 \hline
\end{tabular}
\end{center}
\end{table}

\begin{table}[htp!]
\caption{\footnotesize Average values of CP and AL of 95\% credible intervals based on the proposed extremely-heavy tailed distribution with fixed $\gamma$ (EH) and estimated gamma (aEH), log-Pareto normal distribution with $\rho=0.9$ (LP1) and $\rho=0.7$ (LP2), Cauchy distribution (C), $t$-distribution with 3 degrees of freedom (T3) and estimated degrees of freedom (T), based on 500 replications in 9 combinations of $(100\omega, \mu)$ with $p=10$.
All values are multiplied by 100.
\label{tab:sim-ap2}}
\begin{center}
\begin{tabular}{ccccccccccccc}
\hline
 & $(100\omega, \mu)$ & & EH & aEH & LP1 & LP2 & C & T3 & T & MT & N \\
 \hline
  & (0, --) &  & 94.4 & 94.4 & 92.7 & 84.8 & 87.5 & 92.7 & 93.7 & 94.5 & 94.7 \\
 & (5, 5) &  & 94.4 & 94.3 & 93.2 & 85.7 & 89.3 & 94.3 & 95.2 & 94.6 & 87.7 \\
 & (10, 5) &  & 93.4 & 92.6 & 92.5 & 86.4 & 90.5 & 93.2 & 93.1 & 93.7 & 86.2 \\
 & (5, 10) &  & 95.0 & 94.9 & 93.9 & 85.8 & 89.5 & 95.3 & 97.5 & 95.0 & 86.2 \\
CP & (10, 10) &  & 94.3 & 93.8 & 94.6 & 86.7 & 91.2 & 96.5 & 97.5 & 94.5 & 86.0 \\
 & (5, 15) &  & 94.8 & 94.3 & 93.4 & 85.5 & 90.1 & 95.4 & 98.2 & 94.6 & 86.2 \\
 & (10, 15) &  & 94.2 & 94.1 & 94.5 & 86.3 & 91.1 & 97.2 & 98.5 & 94.6 & 85.7 \\
 & (5, 20) &  & 94.7 & 94.4 & 94.1 & 86.0 & 89.8 & 95.7 & 98.7 & 95.0 & 86.2 \\
 & (10, 20) &  & 94.6 & 94.3 & 94.3 & 86.4 & 91.1 & 97.3 & 99.3 & 94.8 & 86.2 \\
 \hline
 & (0, --) &  & 23.9 & 23.9 & 23.5 & 22.7 & 23.8 & 23.9 & 24.2 & 23.9 & 23.9 \\
 & (5, 5) &  & 25.7 & 25.7 & 25.7 & 24.3 & 25.2 & 26.6 & 29.3 & 25.3 & 35.0 \\
 & (10, 5) &  & 28.3 & 28.2 & 29.9 & 26.8 & 27.1 & 30.6 & 36.4 & 28.4 & 42.9 \\
 & (5, 10) &  & 25.0 & 25.0 & 25.1 & 23.6 & 25.2 & 26.8 & 32.8 & 24.9 & 56.7 \\
AL & (10, 10) &  & 26.2 & 26.3 & 27.0 & 25.1 & 26.9 & 31.1 & 48.2 & 26.5 & 75.1 \\
 & (5, 15) &  & 24.9 & 24.9 & 24.9 & 23.5 & 25.3 & 26.9 & 35.0 & 24.9 & 81.4 \\
 & (10, 15) &  & 26.2 & 26.3 & 26.7 & 24.7 & 27.0 & 31.4 & 58.5 & 26.4 & 109.5 \\
 & (5, 20) &  & 24.9 & 24.9 & 24.8 & 23.4 & 25.3 & 26.9 & 35.8 & 24.8 & 105.8 \\
 & (10, 20) &  & 26.0 & 26.0 & 26.4 & 24.8 & 26.8 & 31.3 & 66.8 & 27.3 & 144.4 \\
 \hline
\end{tabular}
\end{center}
\end{table}

\begin{table}[htp!]
\caption{\footnotesize Computation time (seconds) multiplied by inefficiency factors of the five methods (EH, LP1, T3, MT and N) under four cases of $n$.
\label{tab:sim-cpt}}
\begin{center}
\begin{tabular}{ccccccccccccc}
\hline
&& \multicolumn{4}{c}{$n$}\\
 && 300 & 1200 & 2100 & 3000 \\
 \hline
EH && 17.0 & 73.4 & 132.7 & 198.5 \\
LP1 && 26.2 & 101.1 & 172.1 & 239.4 \\
T3 && 2.5 & 5.0 & 8.0 & 10.9 \\
MT && 2.0 & 4.9 & 7.9 & 10.8 \\
N && 1.1 & 2.0 & 2.9 & 3.9 \\
 \hline
\end{tabular}
\end{center}
\end{table}

\begin{figure}[!htb]
\centering
\includegraphics[width=14cm,clip]{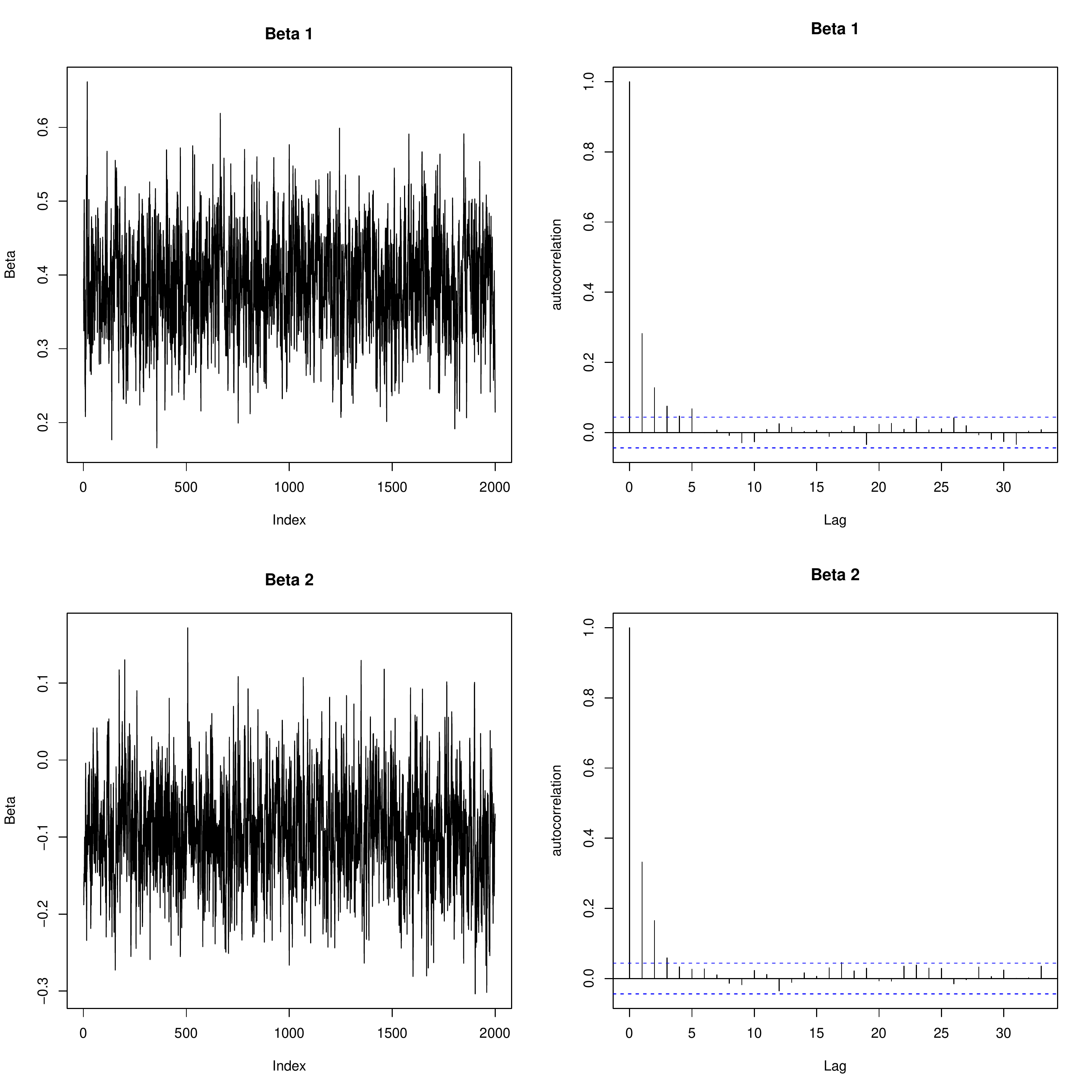}
\caption{Sample paths (Left) and autocorrelation (Right) of the posterior samples of $\beta_2$ and $\beta_3$ in the EH model applied to a simulated data with $p=20$, $\mu=5$ and $\omega=0.05$.
}
\end{figure}

\section*{References}

\begin{itemize}
\item[{[1]}]
Carvalho, C., Polson, N.G. and Scott, J.G. (2010). 
The horseshoe estimator for sparse signals. 
\textit{Biometrika}, \textbf{97}, 465--480. 
\item[{[2]}]
Gagnon, P., Desgagne, P. and Bedard, M. (2019). 
A New Bayesian Approach to Robustness Against Outliers in Linear Regression. 
\textit{Bayesian Analysis}, \textbf{15}, 389--414. 
\item[{[3]}]
Hamura, Y., Irie, K. and Sugasawa, S. (2020). 
Shrinkage with robustness: Log-adjusted priors for sparse signals. 
\textit{arXiv preprint arXiv:2001.08465}
\end{itemize}

\end{document}